\documentclass[12pt,oneside,english]{article}
\usepackage{babel}
\usepackage[T1]{fontenc}
\usepackage[latin9]{inputenc}
\usepackage{simpler-wick}
\usepackage{geometry}
\usepackage{amstext}
\usepackage{amsthm}
\usepackage[justification=centering]{caption}
\usepackage{amssymb}
\usepackage{stmaryrd}
\usepackage{graphicx}
\usepackage{mathtools}
\usepackage{wasysym}
\usepackage{mathrsfs}
\usepackage{bbm}
\usepackage{esint}
\usepackage{dsfont}
\usepackage{xcolor}
\usepackage{verbatim}
\usepackage{subfigure}
\usepackage{setspace}
\usepackage{url}
\usepackage[percent]{overpic}
\usepackage{tikz}
\usepackage{relsize}
\usepackage{titlesec}

\makeatletter

\definecolor{mygray)}{rgb}{0.75, 0.75, 0.75}
\numberwithin{equation}{section}
\numberwithin{figure}{section}
\newtheorem{thm}{Theorem}[section]
\newtheorem{lemma}[thm]{Lemma}
\newtheorem{prop}[thm]{Proposition}
\newtheorem{coro}{Corollary}[thm]

\theoremstyle{definition}

\theoremstyle{remark}
\newtheorem{rmk}{Remark}[section]
\newtheorem{ex}[rmk]{Example}

\numberwithin{thm}{section}

\definecolor{kallecol}{rgb}{.99,.1,.5}
\definecolor{davidcol}{rgb}{.5,.1,.99}
\definecolor{sketchcol}{rgb}{.4,.4,.8}
\definecolor{outlinecol}{rgb}{.8,.4,.3}

\newcommand{\lftsym}{\mathrm{L}}
\newcommand{\rgtsym}{\mathrm{R}}

\newcommand{\epar}{\vspace{3pt}}
\newcommand{\Grass}[1]{\Lambda(#1)}

\newcommand{\cG}{\mathcal{G}}

\newcommand{\Ver}{\mathcal{V}}
\newcommand{\Gra}{\mathcal{G}}
\newcommand{\Edg}{\mathcal{E}}

\newcommand{\Gen}{\textrm{SyFer}_\Gra}
\newcommand{\Action}{\mathcal{S}}
\newcommand{\ParFun}{\mathcal{Z}}
\newcommand{\Dim}[1]{\mathfrak{D}(#1)}
\newcommand{\Dimm}[1]{\mathfrak{D}^2(#1)}

\newcommand{\bigCorFun}[1]{\big\langle #1 \big\rangle}
\newcommand{\BigCorFun}[1]{\Big\langle #1 \Big\rangle}

\newcommand{\Prob}{\mathbb{P}}
\newcommand{\Odd}[1]{\mathcal{E}_{\textnormal{O}}(#1)}
\newcommand{\Even}[1]{\mathcal{E}_{\textnormal{E}}(#1)}

\newcommand{\White}{\mathcal{V}_\white}
\newcommand{\Gray}{\mathcal{V}_\gray}

\newcommand{\Black}{\mathcal{V}_\black}
\newcommand{\ZWhite}{\mathbb{Z}^2_\white}
\newcommand{\ZBlack}{\mathbb{Z}^2_\black}
\newcommand{\ZGray}{\mathbb{Z}^2_\gray}

\newcommand{\PathFact}[1]{\Xi(#1)}
\newcommand{\CompWith}{\sqsubset}
\newcommand{\WithComp}{\sqsupset}
\newcommand{\FromTo}{\rightsquigarrow}
\newcommand{\Expect}{\mathbb{E}}
\newcommand{\Length}[1]{\ell(#1)}

\newcommand{\bx}{\textbf{x}}
\newcommand{\by}{\textbf{y}}

\newcommand{\mtrx}[1]{\textbf{#1}}

\newcommand{\Tail}[2]{\textnormal{tail}_{#2}(#1)}
\newcommand{\Head}[2]{\textnormal{head}_{#2}(#1)}

\newcommand{\SymDiff}{\triangle}

\newcommand{\ev}{\textnormal{ev}}
\newcommand{\LocFi}{\mathcal{F}_{\textnormal{loc}}}
\newcommand{\NuFi}{\mathcal{F}_{\textnormal{null}}}

\newcommand{\rad}{\textnormal{rad}}
\newcommand{\interior}{\textnormal{int}^\sharp}
\newcommand{\Kast}{\mathbf{K}}
\newcommand{\Ball}{\textnormal{B}^\sharp}

\newcommand{\black}{\bullet}
\newcommand{\Fields}{\mathcal{F}}
\newcommand{\NormOrd}[2]{\overset{#1}{\overleftrightarrow{#2}}}
\newcommand{\One}{\mathbbm{1}}
\newcommand{\odd}{\One_\textnormal{odd}}
\newcommand{\even}{\One_\textnormal{even}}
\newcommand{\norm}[1]{\Vert #1\Vert}

\newcommand{\NullRad}[1]{R^\mathcal{N}_{#1}}
\newcommand{\SingRad}[1]{R^\mathcal{S}_{#1}}

\newcommand{\Green}{\mathbf{G}}
\newcommand{\supp}{\textnormal{supp}^\sharp}
\newcommand{\Pierced}[1]{{#1}_{\smallsink}}
\newcommand{\white}{{\circ}}
\newcommand{\gray}{{\scriptscriptstyle\mathbf{1}}}
\newcommand{\sink}{\tikz{\draw[black,scale=0.05] (2.45,2.45) circle (70pt) ;
\draw[thin,scale=0.05] (0,0) -- (4.9,4.9) (0,4.9) -- (4.9,0);}}
\newcommand{\protectsink}{\protect\tikz{\protect\draw[black,scale=0.05] (2.45,2.45) circle (70pt) ;
\protect\draw[thin,scale=0.05] (0,0) -- (4.9,4.9) (0,4.9) -- (4.9,0);}}
\newcommand{\smallsink}{\tikz{\draw[very thin,black,scale=0.03] (2.45,2.45) circle (70pt) ; \draw[very thin,scale=0.03] (0,0) -- (4.9,4.9) (0,4.9) -- (4.9,0);}}
\newcommand{\tinysink}{\tikz{\draw[ultra thin,black,scale=0.019] (2.45,2.45) circle (70pt) ; \draw[ultra thin,scale=0.019] (0,0) -- (4.9,4.9) (0,4.9) -- (4.9,0);}}
\newcommand{\protectsmallsink}{\protect\tikz{\protect\draw[black,scale=0.03] (2.45,2.45) circle (70pt) ;\protect\draw[very thin,scale=0.03] (0,0) -- (4.9,4.9) (0,4.9) -- (4.9,0);}}

\newcommand{\Vir}{\mathbf{Vir}}
\newcommand{\nL}{\textnormal{L}}
\newcommand{\nC}{\textnormal{C}}

\newcommand{\ii}{\mathbbm{i}}
\newcommand{\re}{\textnormal{Re}\,}
\newcommand{\im}{\textnormal{Im}\,}
\newcommand{\sgn}{\mathrm{sgn}}

\newcommand{\bdry}{\partial}

\newcommand{\cconj}[1]{\overline{#1}}

\DeclareMathOperator{\sign}{sgn}

\newcommand{\SymmGrp}{\mathfrak{S}}

\newcommand{\dist}{\mathrm{dist}}

\newcommand{\id}{\mathsf{id}}

\newcommand{\indicator}[1]{\mathbbm{1}_{#1}}

\newcommand{\C}{\mathbb{C}} % complex plane
\newcommand{\Z}{\mathbb{Z}} % integers
\newcommand{\Znn}{\Z_{\geq 0}} % non-negative integers
\newcommand{\Zpos}{\Z_{> 0}} % positive integers
\newcommand{\N}{\mathbb{N}} % natural numbers
\newcommand{\dcoint}[1]{\sqint_{{#1}}} 
\newcommand{\dd}[1]{\ud^\sharp {#1}}
\newcommand{\ud}{\mathrm{d}} % upright d for ordinary differential
\newcommand{\widthL}{\ell_\lftsym}%{\ell^\lftsym}%{\delta}
\newcommand{\widthR}{\ell_\rgtsym}%{\ell^\rgtsym}%{\delta}
\newcommand{\deebar}{{\bar{\partial}}}

\newcommand{\dee}{{\partial}}%{\partial}

\newcommand{\set}[1]{\left\{ #1 \right\}}

\newcommand{\geneigval}{\mu}%{\boldsymbol{\mu}}

\newcommand{\vaceigvalLR}%
	{\geneigval^{(\widthL,\widthR)}_{\emptyset,\emptyset}}

\titleformat{\subsection}[runin]
{\normalfont\bfseries}{\thesubsection}{0pt}{.$\,$}
\titleformat{\section}
{\normalfont\Large\bfseries}{\thesection}{0pt}{.$\ \ $}

\begin{document}

\title{\Large\scshape\bfseries Discrete symplectic fermions on double dimers and their Virasoro representation\vspace{0.5cm}}

\author{{David Adame$\,$-$\,$Carrillo}\vspace{0.2cm}}

%\address{Department of Mathematics and Systems Analysis, Aalto University, P.O. Box 11100, FI-00076 Aalto, Finland}

%\email{david.adamecarrillo@aalto.fi}
\date{\textit{\normalsize Department of Mathematics and Systems Analysis, Aalto University, Espoo, Finland}}
\maketitle

\begin{abstract}
\quad A discrete version of the Conformal Field Theory of symplectic fermions is introduced and discussed. 
Specifically, discrete symplectic fermions are realised as holo\-morphic observables in the double-dimer model.
Using techniques of discrete complex a\-na\-lysis, the space of local fields of discrete symplectic fermions on the square lattice is proven to carry a representation of the Virasoro algebra with central charge $-2$.
\end{abstract}

\section{Introduction}%
\label{sec:intro}
\noindent\qquad
Over the last twenty-five years, numerous conformally invariant properties of the scaling limit of various lattice models have been rigorously established.
Nevertheless, such conformally invariant behaviour of statistical models had been studied in the Physics literature using Conformal Field Theory (CFT) ever since the founding works of Belavin, Po\-lya\-kov and Zamolodchikov in the 1980s \cite{BPZ1,BPZ2}.
Although it is a non-rigorous approach \linebreak to statistical mechanics, CFT has served as a plentiful source of insights to the Ma\-the\-ma\-tics community.
Yet, it is fair to say that it remains far from well-under\-stood from a ma\-the\-ma-\linebreak tical perspective.
\epar

\noindent\qquad
A relevant example of a statistical model is the \emph{dimer model}, in which one takes perfect matchings of the ver\-tices in a graph uniformly at random.
This model has been studied in the Physics literature since as early as 1937 \cite{Fowler};
and, in 1961, Kasteleyn \cite{Kast} and, independently, Fisher and Temperley \cite{Fisher} exactly solved the model in a statistical sense, i.e. they found an exact formula for the number of dimer configurations in finite subgraphs of the square lattice.
As for the scaling limit of the model and its conformal invariance, in the early 2000s, Kenyon established the convergence of the height function of dimers to the Gaussian Free Field \cite{Kenyon1,Kenyon2}.
It is worth pointing out that other approaches have been taken to prove such convergence results \cite{Bufetov,Berestycki}, and that similar results have been proven with more generality \cite{Russkikh}.
\epar

\noindent\qquad
The model considered in this paper is the \emph{double-dimer model}, in which one takes two independent copies of the dimer model.
A strong motivation to study this model is a conjecture by Kenyon, who predicted the loops that arise when superimposing two dimer covers on the square lattice to converge to CLE$_4$ in the scaling limit \cite{RohSch}.
There have been very relevant developments in this direction in the last ten years \cite{Kenyon3,Dubedat,BasChe,BaiWan}, though some questions about the full conjecture are still open \cite{BasChe}.
\epar

\noindent\qquad
In two-di\-men\-sional CFT, the conformal symmetries of the theory are encoded in the infinite-dimensional Lie algebra $\Vir\coloneqq\bigoplus_{n\in\Z}\C \nL_n\oplus\C \nC$ with Lie brackets
\begin{align*}
	\big[\nL_n,\nL_m\big] & =(n-m)\,\nL_{n+m}+\frac{n^3-n}{12}\,\delta_{n+m}\,\nC\,,
	\\
	\big[\nC,\Vir\big]\, & = \ 0\,,
\end{align*}
known as the \emph{Virasoro algebra}.
Algebraically, two-dimensional CFTs are studied in terms of representations of the Virasoro algebra.
In a CFT, the operator $\nC$ is proportional to the identity operator, and its eigenvalue is called the \emph{central charge} of the CFT in question.
\epar

\noindent\qquad
The question that is addressed here is whether the algebraic structure of a CFT can be found already at the lattice level.
In other words, can one build a Virasoro representation using observables on the lattice before taking scaling limits?
This question was answered positively for the discrete Gaussian Free Field and the Ising model in \cite{HKV}.
\epar

\noindent\qquad
Another important CFT shall be considered here: Symplectic fermions, which has central charge $-2$ \cite{Guraire,Kausch1,Kausch2,GaKa}.
This theory is more exotic in the sense that it is of \emph{lo\-ga\-rith\-mic type} (logCFT).
The terminology stems from the fact that logCFTs possess correlation functions with logarithmic dependencies.
On the algebraic side, logCFTs feature more intricate representations of the Virasoro algebra --- representations in which the $\nL_0$ operator cannot be diagonalised.
In particular, symplectic fermions exhibit a subrepresentation \cite{Kausch2} known as a \emph{staggered module} \cite{KytRid}.
\epar

\noindent\qquad
In order to study symplectic fermions on the lattice, a novel discretisation of the theory is introduced.
The fermionic fields $\xi$ and $\eta$ can be defined as holomorphic observables on the double-dimer model on general bi\-par\-tite finite planar graphs.
In particular, one can make precise sense of random variables of the form
$$
\eta(w_1)\xi(b_1)\eta(w_2)\xi(b_2)\cdots\eta(w_n)\xi(b_n)\,,
$$
where $w_i$ and $b_i$ are vertices of different color of the underlying bipartite graph --- see Section \ref{sec: finite} for more details.
The sui\-ta\-bility of this discretisation can be justified by considering the in\-finite volume limit, in which the observables behave in the way predicted by CFT.
The connection between symplectic fermions and double-dimers, although unprecedented, is coherent with the dimer model being studied as CFT of central charge $-2$ \cite{IPRH}.
\epar

\noindent\qquad
The main result of this paper ---Theorem \ref{Virasoro}--- can be informally stated as follows:

\textbf{Theorem.} \textit{The space of local fields of the discrete symplectic fermions on the square lattice constitutes a representation of the Virasoro algebra with central charge $-2$, where the generators $\mathbf{L}_n$ for $n\in\Z$ are defined via a Sugawara construction on the current modes of the fermions $\xi$ and $\eta$.}
\epar

\noindent\qquad
Let us elaborate a bit further on the above statement.
The space in which the Virasoro ac\-tion is defined is the space of \emph{local fields}, which are, in a sense, a generalisation of the maps $z\mapsto\eta(z)$ and $z\mapsto\xi(z)$ --- see Section \ref{sec: fields-currents} for precise definitions.
Using the techniques of discrete complex analysis developed in \cite{HKV}, one can translate the holomorphicity of the fermions into an algebraic language.
In particular, one can define the Fourier modes $\eta_n$ and $\xi_n$ for $n\in\Z$ of the symplectic fermions as operators in the space of local fields.
Discrete ho\-lo\-mor\-phicity yields, then, the exact anticommutation relations of the symplectic fermion algebra: 
$$
\big\{\eta_n,\xi_m\big\}
=
\delta_{n+m}\id\,
\mspace{50mu}
\textnormal{and}
\mspace{50mu}
\big\{\eta_n,\eta_m\big\}=\big\{\xi_n,\xi_m\big\}
=
{0}
$$
for $n,m\in\Z$.
Moreover, the operators $\eta_n$ and $\xi_n$ permit the definition of the Virasoro generators $\mathbf{L}_n$ via a Sugawara construction.

\epar

\textbf{Organisation of the paper.}
In Section \ref{sec: finite}, discrete symplectic fermions are defined on double-dimers on any (dimerable) bipartite finite planar graph and their relation with the Kas\-te\-leyn matrix is established.
Moreover, discrete symplectic fermions are proven to be equi\-va\-lently defined using Grassmann-algebra techniques via a discretisation of the action of symplectic fermions in the continuum.
In Section \ref{sec: infinite}, multipoint correlation functions of discrete symplectic fermions are studied along sequences of growing temperleyan domains of the square lattice.
In particular, the $2$-point function is proven to be closely related to the derivative of the discrete full-plane Green's function, and multipoint correlations are obtained from it by Wick's formula.
In Section \ref{sec: fields-currents}, the space of local fields of discrete sym\-plectic fermions on the the infinite square lattice is defined and discussed along with the current modes of the fermions.
Finally, in Section \ref{sec: virasoro}, the Virasoro modes are defined in terms of the current modes, and they are proven to satisfy the Virasoro commutation relations with central charge $-2$.

\vfill
{\bf Acknowledgments:}
First of all, I want to thank Kalle Kyt\"ol\"a for great ideas, discussions and suggestions at all stages of this project.
I am also very grateful to Misha Basok for se\-veral insightful and fruitful discussions about the topics discussed in Section \ref{sec: infinite}.
I also want to thank Shinji Koshida for suggesting how to write Sections \ref{sec: fields-currents} and \ref{sec: virasoro} in a way that is more natural from a VOA perspective.

\noindent\qquad
This work was supported by Academy of Finland grant 346309 (Finnish Centre of Excellence in Randomness and Structures).

\newpage

\section{Discrete symplectic fermions in finite domains}
\label{sec: finite}
\subsection{Simple paths and fermions.}\label{paths-cycles}
Let $\Gra=(\Ver,\Edg)$ be a bipartite finite planar graph partitioned into \emph{black} vertices $\Black$ and \emph{white} vertices $\White$
that admits at least one \emph{dimer cover}, i.e. there exists a subset of edges $\omega\subset\Edg$,
such that every vertex appears exactly once in $\omega$.
Then, let $\Dim{\Gra{}}$ denote the set of dimer covers on $\cG$, and $\Dimm{\Gra{}}\coloneqq\Dim{\Gra{}}\times\Dim{\Gra{}}$ is the set of \emph{double-dimer covers} on $\Gra$.
See Figure \ref{Definitions} for an example on a subgraph of the square lattice.
The (finite) set $\Dimm{\Gra{}}$ is regarded as a probability space equipped with the uni\-\linebreak form probability measure $\Prob_\Gra$, the expectation value with respect to which is denoted by $\Expect_\Gra$.
\vspace{-10pt}

A famous result by Kasteleyn {\cite{Kast}} is $\vert\Dim{\Gra}\vert=\vert\det \Kast\,\vert$, where $\Kast$ is any \emph{Kasteleyn matrix} of $\Gra$, i.e. $\Kast\colon\Black\times\White\longrightarrow \mathbb{S}^1\cup\set{0}\subset\C$ such that
$\Kast(b,w)\neq 0$ only for $\set{w,b}\in\Edg$, and that satisfies the \emph{Kasteleyn condition}:
Let $(w_1,b_1,\ldots,w_n,b_n,w_{n+1}=w_1)$ be a sequence of consecutively adjacent vertices around a single face of $\Gra$, then
$$
\Kast(b_1,w_1)\cconj{\Kast(b_2,w_1)}\cdots \Kast(b_n,w_n)\cconj{\Kast(b_1,w_n)}=(-1)^{n+1}\,.
$$

\begin{rmk}\label{square_Kast}
Let $\Z^2=\ZWhite\sqcup\ZBlack$ be colored in a chessboard fashion.
If $\Gra=(\Ver,\Edg)$ is the induced (dimerable) graph of a subset $\Ver\subset\Z^2$ with no holes ---see Section \ref{sec: infinite}---,
then, the matrix $\partial:\ZBlack\times\ZWhite\longrightarrow\C$ given by
$\partial_{b,w}\coloneqq\cconj{b-w}$ if $\vert b-w\vert=1$ and $\partial_{b,w}\coloneqq0$ otherwise, is a Kasteleyn matrix of $\Gra$.
\hfill$\diamond$
\end{rmk}

\begin{figure}[h!]
	\centering
	\begin{overpic}[angle=90,scale=0.57, tics=10]{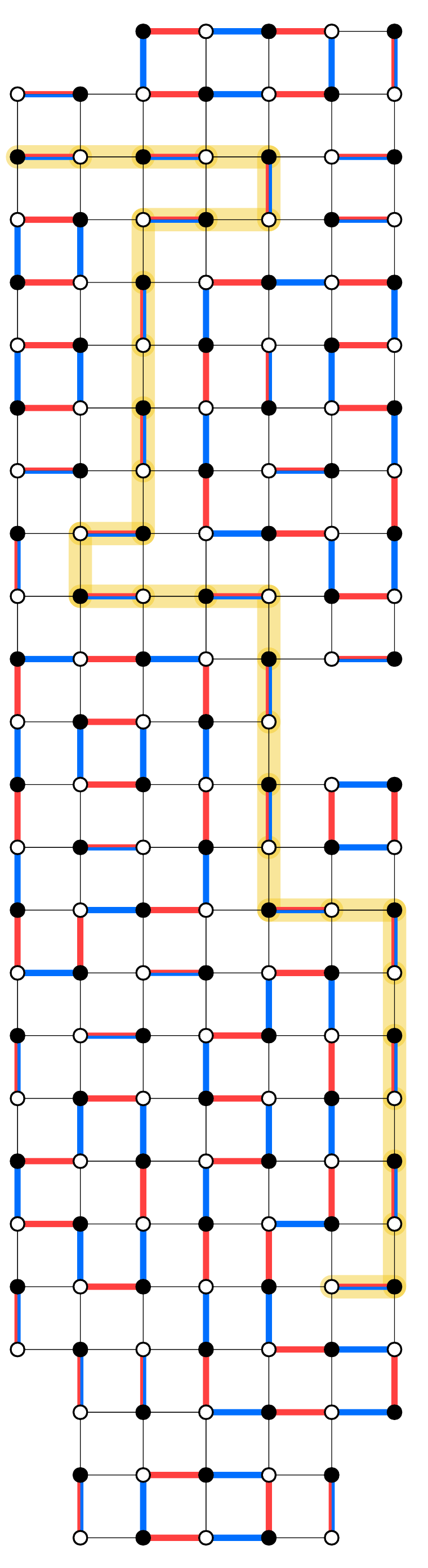}
		\put(46.2,19){$\lambda$}
		\put(11.2,1.6){$b$}
		\put(79.5,19.2){$w$}
	\end{overpic}
	\centering
	\caption{A double-dimer cover $(\omega,\cconj{\omega})$ on a subgraph of the square lattice \\ and an odd simple path $\lambda: w \FromTo b$ adapted to $(\omega,\cconj{\omega})$.}
	\label{Definitions}
\end{figure}

Let us introduce the combinatorial objects that are used to construct symplectic fermions on double dimers.
A \emph{simple path} on $\Gra{}$ is a sequence of distinct consecutively adjacent edges of $\Gra{}$ that does not cross itself,
i.e. $\lambda=(e_1,\dots,e_n)$ with $e_i\in\Edg$,
such that $e_i\cap e_j\neq \emptyset$ if and only if $|i-j|=1$ for $1\leq i\neq j\leq n$.
Then, let $\lambda:x\FromTo y$ denote that $x,y\in\Ver$ are the \emph{endpoints} of $\lambda$, i.e. the vertices that satisfy $x\in e_1$ and $x\notin e_2$, and $y\in e_n$ and $y\notin e_{n-1}$.
Abusing notation, write $e_i\in\lambda$,
and, for $v\in\Ver$, write $v\in\lambda$ if $v\in\bigcup_{i=1}^n e_i$.
Two simple paths $\lambda_1,\lambda_2$ are said not to \emph{intersect} if there is no $v\in\Ver$ satisfying $v\in\lambda_{1}$ and $v\in\lambda_2$, and it is denoted by $\lambda_1\cap\lambda_2=\emptyset$.
Also, $\lambda_1\subset{\lambda_2}$ denotes that $\lambda_1$ is a \emph{subpath} of $\lambda_2$,
i.e. for all $e\in\Edg$ such that $e\in\lambda_1$ it follows that $e\in\lambda_2$.
\epar

An \emph{odd simple path} on $\Gra{}$ is a simple path $\lambda=(e_1,\dots,e_n)$ with $n\in\N$ odd.
Note that the endpoints of odd paths are of different colors.
Let $\Odd{\lambda}\coloneqq \set{e_1,e_3,\dots,e_n}\subset\Edg{}$ and $\Even{\lambda}\coloneqq\set{e_2,e_4,\dots,e_{n-1}}\subset\Edg{}$ denote the set of odd and even edges of $\lambda$ respectively. 
Define also its \emph{odd length} as $\Length{\lambda}\coloneqq\vert\Even{\lambda}\vert$,
and its \emph{path factor} as
\epar
$$
\PathFact{\lambda}\coloneqq (-1)^{\Length{\lambda}}\prod_{e\in\Odd{\lambda}}\Kast(b_e,w_e)\prod_{e\in\Even{\lambda}}\cconj{\Kast(b_e,w_e)}\,,
$$
where $w_e\in\White$ and $b_e\in\Black$ are the white and black vertices of $e=\set{w_e,b_e}$ respectively.
\epar

An odd simple path $\lambda:w\FromTo b$ is said to be \emph{adapted} to a double-dimer cover $(\omega,\cconj{\omega})\in\Dimm{\Gra{}}$ if $\Odd{\lambda}\subset\omega\cap\,\cconj{\omega}$ --- see Figure \ref{Definitions}.
Note that multiple paths with the same endpoints can be adapted to the same double-dimer cover.
Let such adaptedness be denoted by $\lambda\CompWith(\omega,\cconj{\omega})$, and let the indicator function of the event $\set{(\omega,\cconj{\omega})\,\vert\, \lambda\CompWith (\omega,\cconj{\omega})}\subset\Dimm{\Gra{}}$ be denoted by $\indicator{\lambda}$.
\vspace{-10pt}

Then $n$ \emph{pairs of fermions} at $w_1,\ldots,w_n\in\White$ and $b_1,\ldots,b_n\in\Black$ are defined to be the random variable
$$
\eta(w_1)\xi(b_1)\cdots\eta(w_n)\xi(b_n)\coloneqq
\sum_{\sigma\in\SymmGrp_n} \sign\sigma\,
\prod_{i=1}^n\,
\sum_{{\lambda_i:w_i\FromTo b_{\sigma(i)}}}
\PathFact{\lambda_i}\indicator{\lambda_i}\,,
$$
where, for each $i$ and $\sigma$, the second sum runs over the set of simple paths from $w_i$ to $b_{\sigma(i)}$.
\epar

The following proposition states that, in the above sum, terms arising from configurations of $n$ paths with intersections add up to $0$.

\begin{prop}
	\label{disjoint}
	For $w_1,\ldots,w_n\in\White$ and $b_1,\ldots,b_n\in\Black$,
	$$
	\eta(w_1)\xi(b_1)\cdots\eta(w_n)\xi(b_n)
	=
	\sum_{\sigma\in\SymmGrp_n}\sign{\sigma}
	\sum_{\underset{{\lambda_i}\cap{\lambda_j}=\emptyset\ i\neq j\,}{\lambda_i:w_i\FromTo b_{\sigma(i)}}}
	\PathFact{\lambda_1}\cdots\PathFact{\lambda_n}
	\indicator{\lambda_1}\cdots \indicator{\lambda_n}\,,
	$$
	where the second sum runs over the set of $n$ simple paths from $w_i$ to $b_{\sigma(i)}$ that do not intersect each other.
\end{prop}
\begin{proof}
	Fix a double-dimer cover $(\omega,\cconj{\omega})$ throughout the proof.
	By its definition,
	\begin{equation}\label{n-pairs_ev}
	\begin{split}
\big[\eta(w_1)\xi(b_1)\cdots\eta(w_n)\xi(b_n)\big](\omega,\cconj{\omega})&=\\
\sum_{\sigma\in\SymmGrp_n}\sum_{{\lambda_i:w_i\FromTo b_{\sigma(i)}}} &
\sign{\sigma}\ 
\PathFact{\lambda_1}\cdots\PathFact{\lambda_n}\,
\indicator{\lambda_1}(\omega,\cconj{\omega})\cdots \indicator{\lambda_n}(\omega,\cconj{\omega})
	\end{split}
	\end{equation}
	with no restrictions on the sum over paths $\lambda_1,\ldots,\lambda_n$. 
	There is one non-vanishing term in the sum on the right-hand side of Equation (\ref{n-pairs_ev}) for every element of the set
	$$
	S\coloneqq
	\big\{(\sigma;\lambda_1,\ldots,\lambda_n)\,\big\vert\, \sigma\in\SymmGrp_n,\ (\omega,\cconj{\omega})\WithComp\lambda_i:w_i\FromTo b_{\sigma(i)}\textnormal{ for }1\leq i\leq n\big\}\,.
	$$
	Let us build an involution $\iota$ on $S$.
	Say that an element $(\sigma;\lambda_1,\ldots,\lambda_n)$ is \emph{non-intersecting} if it satisfies ${\lambda_i}\cap{\lambda_j}=\emptyset$ for all $1\leq i\neq j\leq n$.
	Then $\iota$ maps non-intersecting elements to themselves.
	Now consider an element $(\sigma;\lambda_1,\ldots,\lambda_n)$ that is not non-intersecting.
	Take $i_1$ as the smallest $1\leq i\leq n$ such that ${\lambda_{i_1}}\cap{\lambda_j}\neq\emptyset$ for some $i_1< j\leq n$, and let $J$ be the set of such $j$'s.
	Let $\Head{\lambda_{j}}{i}:w_j\FromTo x$ be the longest subpath of $\lambda_{j}$ that satisfies that $x\in{\lambda_i}$;
	and take
	$$
	i_2\coloneqq\min\big\{ j\in J\ \big\vert\ \forall k\in J,\  \,\Head{\lambda_{i_1}}{k}\subset\Head{\lambda_{i_1}}{j}\big\}\,.
	$$
	Similarly, for $\lambda_i\cap\lambda_j\neq\emptyset$, let $\Tail{\lambda_{j}}{i}\colon x\FromTo b_{\sigma(j)}$ be the shortest subpath of $\lambda_{j}$ that satisfies $x\in\lambda_{i}$.
	Note it is an empty path when $b_{\sigma(j)}\in\lambda_i$.
	Define, then, $\tilde{\sigma}\coloneqq\sigma\circ[i_1\,i_2]$,
	where $[i_1\,i_2]\in\SymmGrp_n$ is the permutation of $i_1$ and $i_2$,
	and the odd simple paths $\tilde{\lambda}_{i_1}\colon w_{i_1}\FromTo b_{\tilde{\sigma}(i_1)}$ and $\tilde{\lambda}_{i_2}\colon w_{i_2}\FromTo b_{\tilde{\sigma}(i_2)}$ by swapping the tails of $\lambda_{i_1}$ and $\lambda_{i_2}$ around --- see Figure \ref{Invo_1}.
	Note that, since $\lambda_{i_1},\lambda_{i_2}\CompWith(\omega,\cconj{\omega})$, necessarily $\tilde{\lambda}_{i_1},\tilde{\lambda}_{i_2}\CompWith(\omega,\cconj{\omega})$.
	Then, take
	$$
	\iota(\sigma;\lambda_1,\ldots,\lambda_{i_1},\ldots,\lambda_{i_2},\ldots,\lambda_n)=(\tilde{\sigma};\lambda_1,\ldots,\tilde{\lambda}_{i_1},\ldots,\tilde{\lambda}_{i_2},\ldots,\lambda_n)\,.
	$$
	Indeed, $\iota\circ \iota=\id$.
	The result follows by observing that $\PathFact{\lambda_{i_1}}\,\PathFact{\lambda_{i_2}}=\PathFact{\tilde{\lambda}_{i_1}}\,\PathFact{\tilde{\lambda}_{i_2}}$ and $\sign{\sigma}=-\sign{\tilde{\sigma}}$,
	i.e. the terms in the sum in Equation (\ref{n-pairs_ev}) not arising from non-intersecting elements of $S$ cancel out pairwise.
\end{proof}
\begin{figure}[t!]
	\centering
	\begin{overpic}[angle=90,scale=0.57, tics=10]{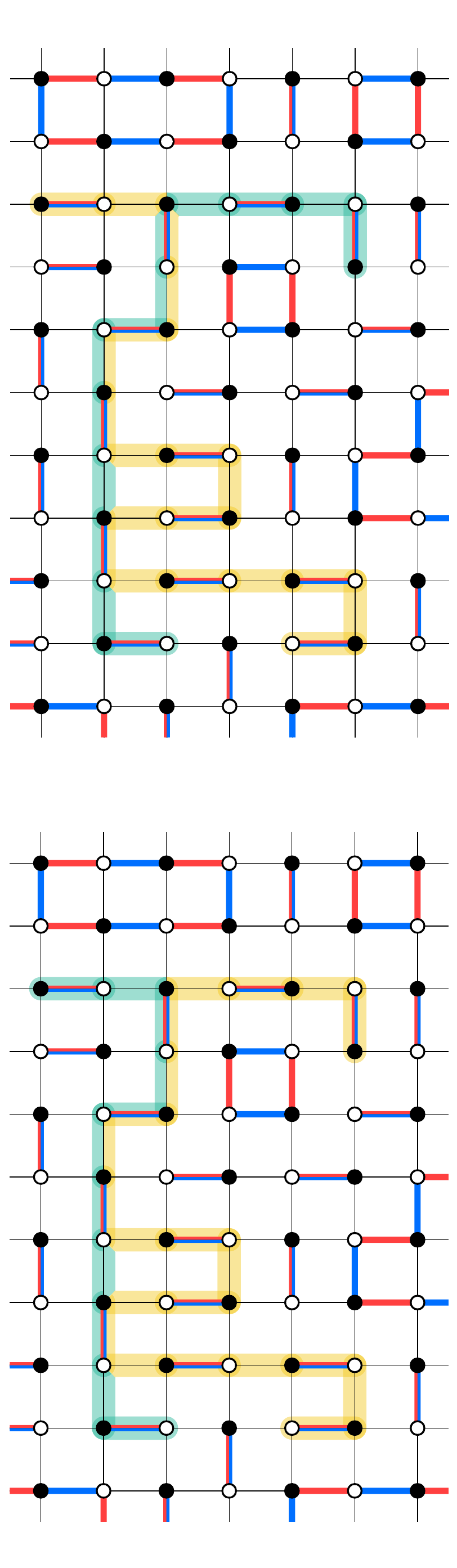}
		\put(30,4){\footnotesize $\lambda_{i_1}$}
		\put(25.7,15.7){\footnotesize $\lambda_{i_2}$}
		\put(48.3,15.8){$\longmapsto$}
		\put(48.3,13.3){\rotatebox{180}{$\longmapsto$}}
		\put(49.8,18){$\iota$}
		\put(75.7,15.7){\footnotesize $\tilde{\lambda}_{i_2}$}
		\put(80,3.6){\footnotesize $\tilde{\lambda}_{i_1}$}
		\put(41.5,11.6){\footnotesize $w_{i_1}$}
		\put(41.5,16.8){\footnotesize $w_{i_2}$}
		\put(91.5,11.6){\footnotesize $w_{i_1}$}
		\put(91.5,16.8){\footnotesize $w_{i_2}$}
		\put(15,24.2){\footnotesize $b_{\sigma(i_1)}$}
		\put(11,0.6){\footnotesize $b_{\sigma(i_2)}$}
		\put(65.05,24.2){\footnotesize $b_{\tilde{\sigma}(i_2)}$}
		\put(61.05,0.6){\footnotesize $b_{\tilde{\sigma}(i_1)}$}
	\end{overpic}
	\caption{\centering Involution $\iota$ in the proof of Proposition \ref{disjoint}.}
	\label{Invo_1}
\end{figure}

\begin{coro}\label{null-fermions}
	If $w_i=w_j$ or $b_i=b_j$ for some $i\neq j$, then $\eta(w_1)\xi(b_1)\cdots\eta(w_n)\xi(b_n)\equiv 0$.
\end{coro}
\begin{proof}
	If not all $w_1,\ldots,w_n$ and $b_1,\ldots,b_n$ are distinct, there are no non-intersecting elements in $S$, so the whole sum adds up to $0$.
\end{proof}

\subsection{Discrete holomorphicity of fermions.} 
On $\Gra{}$, the Kasteleyn matrix provides a notion of differentiation:
Let $V$ be a vector space --- usually $\C$.
For $f\colon\Ver\longrightarrow V$,
the $V$-valued function $\dee^\Kast f$ is defined, on black and white vertices respectively, by
\epar
\epar
$$
\dee^\Kast f(b)\coloneqq\sum_{\underset{\set{b,w}\in\Edg{}}{w\in\White}}\Kast(b,w)f(w)
\mspace{50mu}
\textnormal{and}
\mspace{50mu}
\dee^\Kast f(w)\coloneqq-\sum_{\underset{\set{b,w}\in\Edg{}}{b\in\Black}}\Kast(b,w)f(b)\,.
$$
Similarly, $\deebar^\Kast f$ is defined by
$$
\deebar^\Kast f(b)\coloneqq\sum_{\underset{\set{b,w}\in\Edg{}}{w\in\White}}\cconj{\Kast(b,w)}f(w)
\mspace{50mu}
\textnormal{and}
\mspace{50mu}
\deebar^\Kast f(w)\coloneqq-\sum_{\underset{\set{b,w}\in\Edg{}}{b\in\Black}}\cconj{\Kast(b,w)}f(b)\,.
$$

For an expression depending on of several vertices, the notation $\dee_x^\Kast$ and $\deebar_x^\Kast$ is used to clarify with respect to which variable the $\Kast$-derivative is taken.
\epar

The following result proves $\Expect_\Gra[\eta(\cdot)\xi(\cdot)]$ to be (the complex conjugate of) the non-zero entries of the \emph{coupling function} ---in Kenyon's sense \cite{Kenyon1}--- of double dimers on $\Gra$,
i.e. $\Expect_\Gra[\eta(w)\xi(b)]=\cconj{\Kast^{-1}(w,b)}$.
A purely combinatorial proof shall be presented here, although a shorter proof can be written using Grassmann algebra techniques --- see Subsections \ref{GrassandWick} and \ref{SFandDD}.

\begin{thm}\label{discrete_holomorphicity}
Fix $w_0\in\White$ and $b_0\in\Black$. Then,
\epar
$$
\deebar^\Kast_w\Expect_\Gra\big[\,\eta(w_0)\xi(w)\,\big]=\delta_{w,w_0}
\mspace{50mu}
\textit{and}
\mspace{50mu}
\deebar^\Kast_b\Expect_\Gra\big[\,\eta(b)\xi(b_0)\,\big]=-\delta_{b,b_0}\,.
$$

for all $w\in\White$ and $b\in\Black$.
\end{thm}

\begin{figure}[b!]
	\centering
	\begin{overpic}[scale=0.75, tics=10]{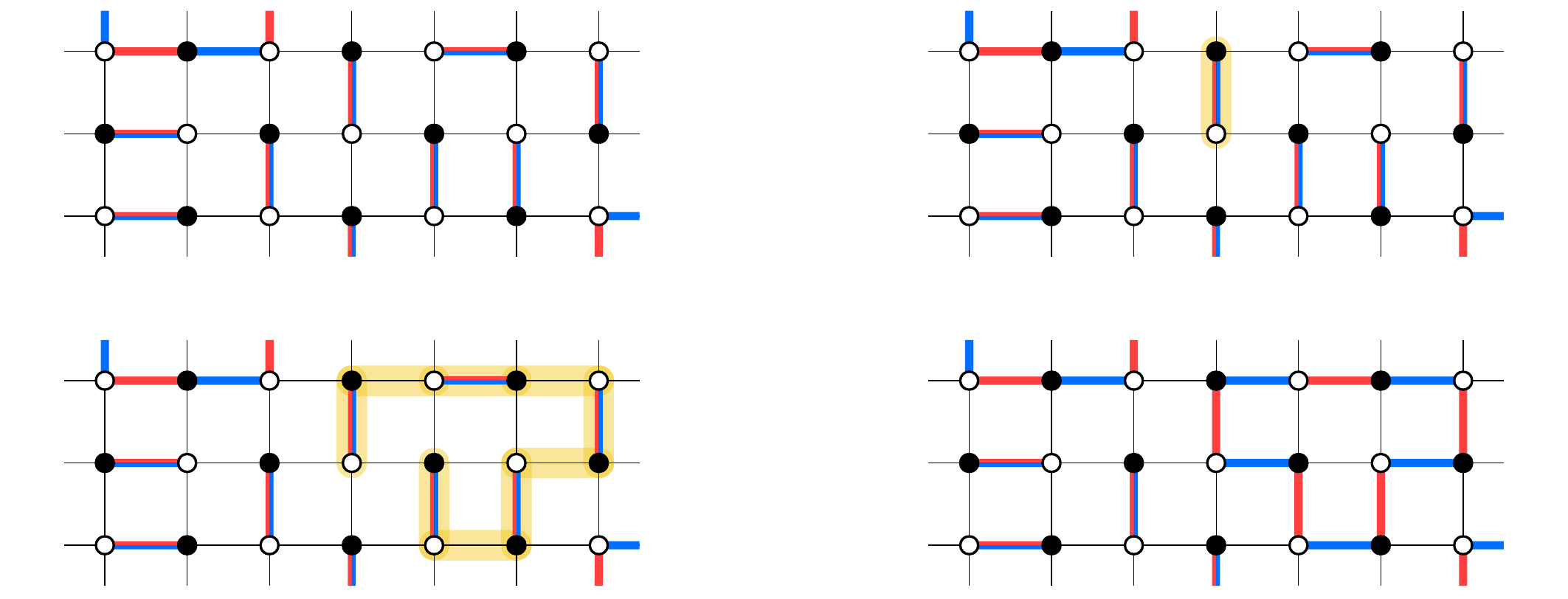}
		\put(49.7,24.5){$\phi$}
		\put(47.7,31){{\Large$\longmapsto$}}
		\put(47.5,28.8){\rotatebox{180}{\Large$\longmapsto$}}
		\put(49.7,33.6){$\phi^{-1}$}
		\put(49.7,12.8){$\phi$}
		\put(47.7,10){{\Large$\longmapsto$}}
		\put(47.5,7.8){\rotatebox{180}{\Large$\longmapsto$}}
		\put(49.7,3.2){$\phi^{-1}$}
		\put(19.6,6.5){$w$}
		\put(19.6,27.5){$w$}
		\put(74.6,6.5){$w$}
		\put(74.6,27.5){$w$}
		\put(29,9.6){$b$}
		\put(20.6,36){$b$}
		\put(84,9.6){$b$}
		\put(75.6,36){$b$}
		\put(24.3,15.5){$\lambda$}
		\put(79.3,31.5){$\lambda_w$}
		\put(39,37){$(\omega,\cconj{\omega})$}		\put(53,37){$(\omega_w,\cconj{\omega})$}		\put(39,16){$(\omega,\cconj{\omega})$}		\put(53,16){$(\omega_\lambda,\cconj{\omega})$}
	\end{overpic}
	\caption{The bijection $\phi$ in the proof of Theorem \ref{discrete_holomorphicity}.}
	\label{Bijec_w=w0}
\end{figure}

\begin{proof}
Explicitly,
\begin{equation}\label{deebar_2p}
\deebar^\Kast_w\Expect_\Gra\big[\,\eta(w_0)\xi(w)\,\big]=
\frac{1}{\vert\Dimm{\Gra{}}\vert}\sum_{(\omega,\cconj{\omega})\in\Dimm{\Gra{}}}
\sum_{b\sim w}
\sum_{\lambda:w_0\FromTo b}
\cconj{\Kast(b,w)}\,{\PathFact{\lambda}\indicator{\lambda}(\omega,\cconj{\omega})}\,.
\end{equation}
There is a non-zero term in the above sum for each element of the set $S^{(w_0)}\coloneqq\bigsqcup_{b\sim w}S_b^{(w_0)}$, where
$$
S_b^{(w_0)}\coloneqq \big\{(\omega,\cconj{\omega};\lambda)\,\big\vert\, (\omega,\cconj{\omega})\in\Dimm{\Gra{}},\,(\omega,\cconj{\omega})\WithComp\lambda:w_0\FromTo b\,\big\}\,.
$$
Assume first $w=w_0$, and let us build a bijection $\phi:S^{(w_0)}\longrightarrow\Dimm{\Gra{}}$ --- see Figure \ref{Bijec_w=w0}.

-- For $(\omega,\cconj{\omega};\lambda)\in S_b^{(w)}$ for some $\set{w,b}\eqqcolon e_w\in\Edg{}$, take the dimer cover
$$
\omega_\lambda\coloneqq \big(\omega\setminus\Odd{\lambda}\big)\cup \Even{\lambda}\cup \set{e}\,.
$$
Note if $\Length{\lambda}=0$ then ${\omega}_\lambda={\omega}$.
Set then $\phi(\omega,\cconj{\omega};\lambda)\coloneqq(\omega_\lambda,\cconj{\omega})$.

-- Conversely, for a double-dimer cover $(\omega,\cconj{\omega})\in\Dimm{\Gra{}}$,
let $e_1,\ldots,e_n\in\omega$ and $\bar{e}_1,\ldots,\bar{e}_{n+1}\in\cconj{\omega}$ be the edges that make a simple path
$(\bar{e}_1,e_1,\ldots,e_n,\bar{e}_{n+1})\eqqcolon\lambda_w\colon w \FromTo b$ for some $b\in\Black$ adjacent to $w$.
Take the dimer cover $\omega_w\coloneqq(\omega\setminus\set{e_1,\ldots,e_n})\cup\set{\bar{e}_1,\ldots,\bar{e}_{n+1}}$
and set $\phi^{-1}(\omega,\cconj{\omega})\coloneqq(\omega_w,\cconj{\omega};\lambda_w)$.
Note $n$ can be $0$ and, in that case, $\omega_w=\omega$.

Indeed, $\phi\,\circ\,\phi^{-1}=\id_{\Dimm{\Gra{}}}$.
Moreover, the term  $\cconj{\Kast(b,w)}\,\PathFact{\lambda_w}$ in Equation \eqref{deebar_2p} arising from $\phi^{-1}(\omega,\cconj{\omega})=(\omega_w,\cconj{\omega};\lambda_w)$ is $1$ by virtue of the Kasteleyn condition.
It follows that $\deebar^\Kast_w\Expect_\Gra\big[\eta(w_0)\xi(w)\big]{}=1$ when $w=w_0$.

Assume now $w\neq w_0$, and let us build an involution $\iota$ on $S^{(w_0)}$ --- see Figure \ref{Invo_wneqw0}.
Consider an element $(\omega,\cconj{\omega};\lambda)\in S^{(w_0)}_b\subset S^{(w_0)}$.

-- If $w\in\lambda$, there exists $\{w,b_0\}\in\Even{\lambda}$.
Then, let $\textnormal{head}_\lambda:w_0\FromTo b_0$ and $\textnormal{tail}_\lambda:w \FromTo b$ be subpaths of $\lambda$.
Consider $\omega_\lambda\coloneqq\big(\omega\setminus\Odd{\textnormal{tail}_\lambda}\big)\cup\Even{\textnormal{tail}_\lambda}\cup\big\{\set{w,b}\big\}$ and set $\iota(\omega,\cconj{\omega};\lambda)\coloneqq(\omega_\lambda,\cconj{\omega};\textnormal{head}_\lambda)$.
Note $\omega_\lambda=\omega$ if $\Length{\textnormal{tail}_\lambda}=1$.

-- Otherwise, if $w\notin\lambda$,
let $e_1,\ldots,e_n\in\omega$ and $\bar{e}_1,\ldots,\bar{e}_{n+1}\in\cconj{\omega}$ be the edges that make a simple path
$(\bar{e}_1,e_1,\ldots,e_n,\bar{e}_{n+1})\eqqcolon\lambda_w\colon w \FromTo b_1$ for some $b_1\in\Black$ adjacent to $w$.
Let $\tilde{\lambda}\colon w_0\FromTo b_1$ be the concatenation of $\lambda$ with $(\set{b,w})$ and $\lambda_w$, and consider the dimer cover
$
\omega_w\coloneqq\omega\setminus\big(\Even{\lambda_w}\cup \big\{\set{w,b_1}\big\}\big)\cup\Odd{\lambda_w}\,.
$
In this case, set $\iota(\omega,\cconj{\omega};\lambda)\coloneqq(\omega_w,\cconj{\omega};\tilde \lambda)$.
Note $\omega_w=\omega$ if $n=0$.

\begin{figure}[b!]
	\centering
	\begin{overpic}[scale=0.75, tics=10]{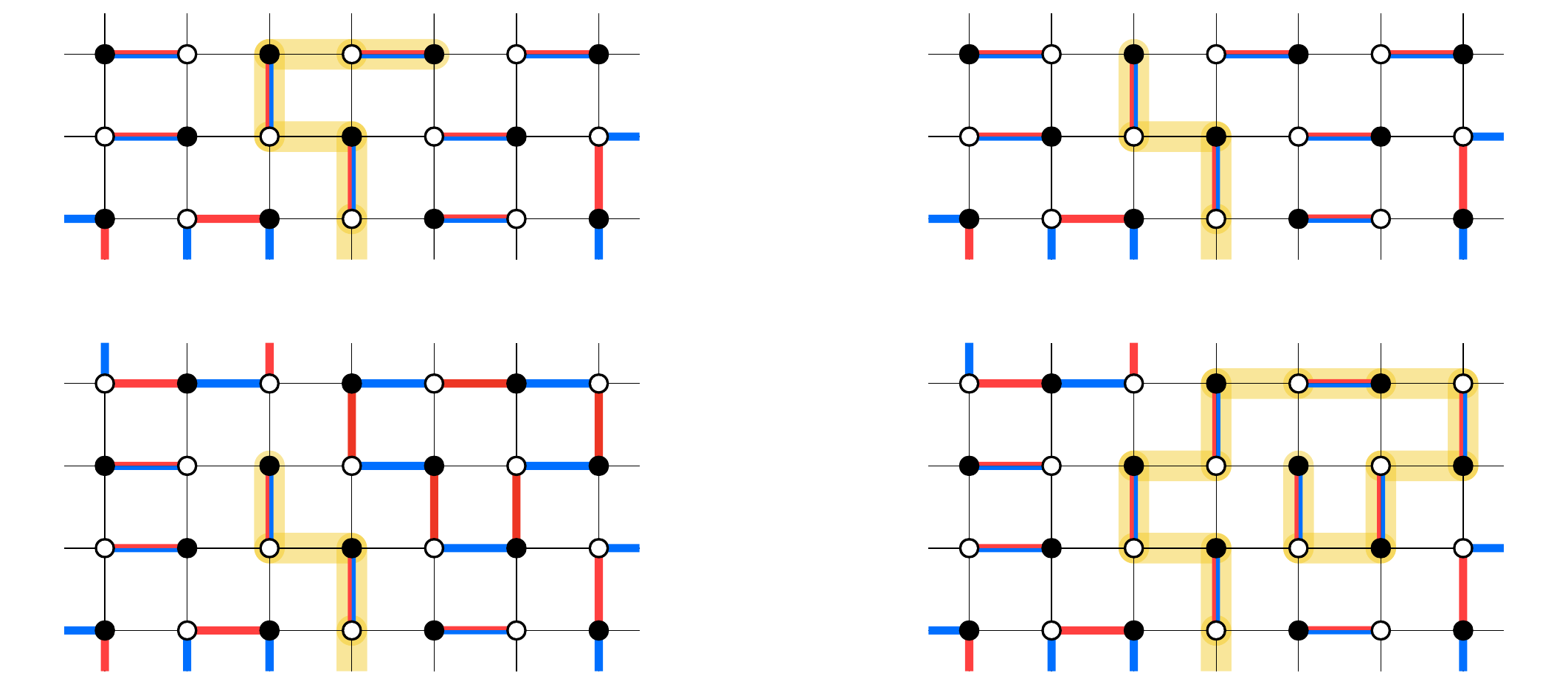}
		\put(50,38.8){$\iota$}
		\put(47.7,36){{\Large$\longmapsto$}}
		\put(47.7,33.8){\rotatebox{180}{\Large$\longmapsto$}}
		\put(50.5,29.2){$\iota$}
		\put(50,14.8){$\iota$}
		\put(47.7,12){{\Large$\longmapsto$}}
		\put(47.7,9.8){\rotatebox{180}{\Large$\longmapsto$}}
		\put(50.5,5.2){$\iota$}
		\put(20,41.3){$w$}
		\put(14.3,41.3){$b_0$}
		\put(28.3,41.1){$b$}
		\put(24,26.5){$\lambda$}
		\put(75,41.3){$w$}
		\put(84,41.1){$b$}
		\put(79,26.5){$\tilde\lambda$}
		\put(23,11.5){$w$}
		\put(14.3,15.1){$b$}
		\put(28.3,14.8){$b_1$}
		\put(24,0){$\lambda$}
		\put(78.3,11.3){$w$}
		\put(70.3,14.7){$b$}
		\put(79,0){$\tilde\lambda$}
		\put(39,42){$(\omega,\cconj{\omega})$}		\put(53,42){$(\omega_\lambda,\cconj{\omega})$}		\put(39,21){$(\omega,\cconj{\omega})$}		\put(53,21){$(\omega_w,\cconj{\omega})$}
	\end{overpic}
	\caption{The involution $\iota$ in the proof of Theorem \ref{discrete_holomorphicity}.}
	\label{Invo_wneqw0}
\end{figure}
Indeed $\iota\,\circ\,\iota=\id_{S^{(w_0)}}$.
Moreover, consider the latter scenario, when $w\notin\lambda$.
The term of $(\omega,\cconj{\omega};\lambda)$ in Equation \eqref{deebar_2p}  is
$$
\cconj{\Kast(b,w)}\,\PathFact{\lambda}
=
\cconj{\Kast(b,w)}(-1)^{\Length{\lambda}}
\prod_{e\in\Odd{\lambda}}\Kast(b_e,w_e)
\prod_{e\in\Even{\lambda}}\cconj{\Kast(b_e,w_e)}\,
$$
whereas the term of $\iota(\omega,\cconj{\omega};\lambda)=(\omega_w,\cconj{\omega};\tilde\lambda)$ is
\begin{align*}
\cconj{\Kast(b_1,w)}\,\PathFact{\tilde\lambda}\ 
= & \ 
\cconj{\Kast(b_1,w)}(-1)^{\Length{\tilde\lambda}}
\prod_{e\in\Odd{\tilde\lambda}}{\Kast(b_e,w_e)}
\prod_{e\in\Even{\tilde\lambda}}\cconj{\Kast(b_e,w_e)}\,
\\
=&\ 
(-1)^{\Length{\lambda}+1+\Length{{\lambda_w}}}\,
\cconj{\Kast(b_1,w)}
\prod_{e\in\Odd{\lambda}}\Kast(b_e,w_e)
\prod_{e\in\Even{\lambda}}\cconj{\Kast(b_e,w_e)}
\\
& \mspace{50mu}\times
\cconj{\Kast(b,w)}
\prod_{e\in\Odd{\lambda_w}}{\Kast(b_e,w_e)}
\prod_{e\in\Even{\lambda_w}}\cconj{\Kast(b_e,w_e)}\,
\\
=&\ 
-\cconj{\Kast(b,w)}\,\PathFact{\lambda} \Bigg[\cconj{\Kast(b_1,w)}\,(-1)^{\ell(\lambda_w)}\prod_{e\in\Odd{\lambda_w}}{\Kast(b_e,w_e)}
\prod_{e\in\Even{\lambda_w}}\cconj{\Kast(b_e,w_e)}\Bigg]\,.
\end{align*}
Again, the factor in square brackets is $1$ by virtue of the Kasteleyn condition.
Therefore, when $w\neq w_0$, all the terms in Equation \eqref{deebar_2p} cancel out pairwise,
i.e. $\deebar^\Kast_w\Expect_\Gra\big[\eta(w_0)\xi(w)\big]{}=0$.
The same arguments can be used to prove $\deebar^\Kast_b\Expect_\Gra\big[\eta(b)\xi(b_0)\big]=-\delta_{b,b_0}$.
\end{proof}

Using combinatorial arguments similar to the coupling $\phi$ in the above proof one can prove the following result, that is stated here as a corollary.

\begin{coro}\label{probabilistic}
	Let $e_1=\set{w_1,b_1},\ldots, e_n=\set{w_{n},b_{n}}\in\Edg$ and be $n$ edges.
	Then,
	$$
	\Expect_\Gra\big[\,\eta(w_1)\xi(b_1)\cdots\eta(w_n)\xi(b_n)\,\big]
	=
	\frac
	{\Expect_\Gra\big[\,\indicator{e_1}\cdots\indicator{e_n}\,\big]}
	{\cconj{\Kast(b_1,w_1)}\cdots\cconj{\Kast(b_n,w_n)}}
	\,
	$$
	where $\indicator{e}$ is the indicator function $e$ being open, i.e. the event $\set{(\omega,\cconj{\omega})\,\vert\,e\in\omega}$.
\end{coro}

\subsection{Grassmann formalism and Wick's theorem.}\label{GrassandWick}
The \emph{Grassmann algebra} $\Grass{G}$ of a finite set of \emph{generators} $G$ is the quotient of the free non-commutative ring $\C\langle G\rangle$ by the two-sided ideal generated by $g_1g_2+g_2g_1$ for all $g_1,g_2\in G$.
Note $\Grass{G}$ is a finite-dimensional algebra over $\C$.
Given an \emph{order of the generators} $\sigma$, i.e. a injective $\sigma:G\longrightarrow \set{1,\dots,|G|}$,
one can construct a basis of $\Grass{G}$ indexed by the subsets of $G$:
For a subset $S\subset G$,
write $S=\set{g_1,g_2\ldots,g_n}$ so that $\sigma(g_1)<\sigma(g_2) \cdots <\sigma(g_n)$
and define $v_{S}^\sigma\coloneqq g_1g_2\cdots g_n\in\Grass{G}$ and $v_{\emptyset}^\sigma\coloneqq 1\in\Grass{G}$.
Then, the set $\set{v_{S}^\sigma}_{S\subset G}$ constitutes a basis of $\Grass{G}$.
\begin{rmk}\label{nilpotent}
In the basis $\set{v_{S}^\sigma}_{S\subset G}$, any element $v\in\Grass{G}$ that has vanishing projection onto the subspace spanned by $v_{\emptyset}^\sigma$ satisfies $v^{|G|+1}=0$.
For such elements, the exponential map ---given by the usual power series--- is well-defined.
\hfill$\diamond$
\end{rmk}
Given an order $\sigma$ of $G$, one can canonically construct the bilinear form $\langle\cdot,\cdot\rangle_\sigma^{}$ determined by $\langle v_{S_1}^\sigma,v_{S_2}^\sigma\rangle_\sigma^{}=\delta_{S_1,S_2}$ for $S_1,S_2\subset G$.
Then, the \emph{Berezin integral} of $v\in\Grass{G}$ with respect to the order $\sigma$ is defined as
$$
\int v\ \ud G_\sigma\coloneqq \langle v_G^\sigma,v\rangle_\sigma^{}.
$$
\vspace{1pt}

Note that, given two orders $\sigma,\varsigma$ of $G$, the Berezin integrals with respect to each of them differ by an overall factor $\sgn(\varsigma\circ\sigma^{-1})$.
\epar

Take the set of generators $G=\set{x_i,y_i}_{i=1}^n$ and a matrix $\mtrx{A}=\set{A_{ij}}_{i,j=1}^n\in\textnormal{GL}_n(\C)$.
Consider the \emph{action}
$$
\Action[\bx,\by]\coloneqq \sum_{i=1}^n\sum_{j=1}^nx_i A_{ij} y_j\in\Grass{G}\,.
$$
For a given order $\sigma$ of $G$, the \emph{partition function} is defined as $\ParFun_{\Action{}}\coloneqq\int e^{\Action[\bx,\by]}\ud G_\sigma$.
Note it is well-defined complex number by Remark \ref{nilpotent}.

\begin{rmk}\label{part_corr}
	Fixing the order to be $\sigma(x_i)=2i-1$ and $\sigma(y_i)=2i$, one gets {\cite{CSS}}
	$$
	\ParFun_{\Action{}}= \det{\mtrx{A}}\,.
	$$
	For $1\leq i_1<\cdots<i_m\leq n$ and $1\leq j_1<\cdots<j_m\leq n$, \emph{correlation functions} are defined by
	$$
	\BigCorFun{\prod_{k=1}^m x_{i_k}y_{j_k}}
	\coloneqq
	\frac{1}{\ParFun_\Action{}}\int\prod_{k=1}^m x_{i_k}y_{j_k}\ e^{\Action[\bx,\by]}\ud G_\sigma
	$$
	and can be proven {\cite{CSS}} to be given by
	$$
	\BigCorFun{\prod_{k=1}^m x_{i_k}y_{j_k}}
	=
	\varepsilon_{IJ} \det \mtrx{A}_{\hat{I}\hat{J}}\,,
	$$
	where $\mtrx{A}_{\hat{I}\hat{J}}$ is the matrix obtained by removing the columns $I=\set{i_1,\ldots,i_m}$ and rows $J=\set{j_1,\ldots,j_m}$ from $\mtrx{A}$, and $\varepsilon_{IJ}=(-1)^{i_1+\cdots+i_m+j_1+\cdots+j_m}$.\hfill$\diamond$
\end{rmk}

Therefrom, one can prove Wick's formula, which is stated here as a proposition.

\begin{prop}[Wick's theorem]\label{WickGrass}
	Let $x_{i_1},\ldots,x_{i_m},y_{j_1},\ldots,y_{j_m}\in G$ be $2m$ generators.
	Then,
	$$
	\bigCorFun{\,x_{i_1}y_{j_i}\cdots x_{i_m}y_{j_m}\,}=
	\sum_{\sigma\in\SymmGrp_m}\sign{\sigma}\,\bigCorFun{x_{i_1}y_{j_{\sigma(1)}}}\cdots \bigCorFun{x_{i_m}y_{j_{\sigma(m)}}}\,.
	$$
\end{prop}

\subsection{Discrete symplectic fermions and double dimers.}\label{SFandDD}
The observables on double dimers described above can be alternatively found using the Grassmann formalism by considering the appropriate discretisation of the continuum action of symplectic fermions {\cite{Kausch2}}.
Let us build such discretisation on a dimerable bipartite graph $\Gra=(\White\sqcup\Black,\Edg)$ equipped with a Kasteleyn matrix $\Kast$.
Then, take the set of generators 
$$
\Gen\coloneqq
\big\{\,
\eta(w),\bar\eta(w),\xi(b),\bar{\xi}(b)
\ \big\vert\ 
w\in\White,\,b\in\Black
\,\big\}\,,
$$
and consider its Grassmann algebra $\Grass{\xi,\eta}\coloneqq\Grass{\Gen}$.
Abusing notation, $\xi$ and $\bar{\xi}$ are viewed as $\Grass{\xi,\eta}$-valued functions on $\Black$
as well as $\eta$ and $\bar{\eta}$ as $\Grass{\xi,\eta}$-valued functions on $\White$,
so that we can apply the operators $\dee^\Kast$ and $\deebar^\Kast$ on them.
For Berezin integration, take any order $\varsigma$ of the vertices, i.e. $\varsigma\colon \Ver\longrightarrow \set{1,\ldots,\vert\Ver\vert}$ injective, and take the order of the generators given by $\xi(b)\longmapsto 2\varsigma(b)$, $\cconj\xi(b)\longmapsto 2\varsigma(b)-1$, $\eta(w)\longmapsto 2\varsigma(w)-1$ and $\cconj\eta(w)\longmapsto 2\varsigma(w)$.
Note any order $\varsigma$ of the vertices yields the same sign for Berezin integration, and let such integration be denoted by $\int\cdot\,\ud\eta\ud\cconj\eta\ud\cconj\xi\ud\xi$.
\epar

The discretised version of the continuum action in {\cite{Kausch2}} given by
$$
\Action[\,\eta,\xi,\cconj\eta,\cconj\xi\,]\coloneqq
\sum_{w\in\White}\Big(\eta(w)\deebar^\Kast\xi(w)+\cconj{\eta}(w)\dee^\Kast\cconj{\xi}(w)\Big)
$$
leads to the same observables as the ones defined on double dimers.

\begin{prop}\label{Equivalence}
	Let $w_1,\ldots,w_n\in\White$ and $b_1,\ldots,b_n\in\Black$ be $2n$ vertices.
	Then
	$$
	\Big\langle{\,\eta(w_1)\xi(b_1)\cdots\eta(w_n)\xi(b_n)\,}\Big\rangle
	=
	\Expect_\Gra\Big[\,\eta(w_1)\xi(b_1)\cdots\eta(w_n)\xi(b_n)\,\Big]\,.
	$$
\end{prop}
\begin{proof}
	By Remark \ref{part_corr}, the partition function is the determinant of a block-diagonal matrix with two blocks: $\Kast$ and $\cconj\Kast$.
	The partition function is then $\det\Kast\,\cconj{\det \Kast}$, which, by virtue of Kasteleyn's theorem, equals $\vert\Dim{\Gra}\vert^2$.
	The action $\Action[\xi,\eta]$ contains the terms $\cconj{\Kast(b,w)}\,\eta(w)\xi(b)$ and $\Kast(b,w)\,\cconj\eta(w)\cconj\xi(b)$ for each edge $\set{b,w}\in\Edg$.
	Therefore, the terms that contribute to $\int \eta(w_1)\xi(b_1)\cdots$ $\eta(w_n)\xi(b_n)e^{\Action[\eta,\xi,\cconj\eta,\cconj\xi]}\,\ud\eta\ud\cconj\eta\ud\cconj\xi\ud\xi$ are in the $(\vert\Ver\vert/2-n)$-th term in the series of the exponential:
	\begin{equation}
		\eta(w_1)\xi(b_1)\cdots\eta(w_n)\xi(b_n)\frac{\Action[\xi,\eta]^{\vert\Ver\vert/2-n}}{(\vert\Ver\vert/2-n)!}\,.
	\end{equation}
	Each monomial appears $(\vert\Ver\vert/2-n)!$ times, and can be identified with a pair consisting of a dimer cover on $\Gra$ ---from the factors $\cconj\eta\cconj\xi$--- and a dimer cover on the graph $\Gra$ with the vertices $w_1,b_1,\ldots,w_n,b_n$ removed.
	That is, there is a term that survives Berezin integration for each configuration consisting of $n$ non-intersecting paths $\lambda_i\colon w_i\FromTo b_{\sigma(i)}$ with $\sigma\in\SymmGrp_n$ for $1\leq i\leq n$, and a double-dimer cover on the complement of those paths in $\Gra$.
	In light of the Kasteleyn condition, the factors arising from loops are always $1$, and the contribution of such term is $\sign{\sigma}\prod_{i=1}^n\PathFact{\lambda_i}$.
	The statement becomes clear by writing the right-hand side using the expression in Proposition \ref{disjoint} and noting one has the same sum as on the left-hand side.
\end{proof}

\begin{prop}[Wick's theorem]\label{WickProb}
	Fix $w_1,\ldots,w_n\in\White$ and $b_1,\ldots,b_n\in\Black$.
	Then,
	\epar
	$$
	\Expect_\Gra\big[\,\eta(w_1)\xi(b_1)\cdots\eta(w_n)\xi(b_n)\,\big]=\sum_{\sigma\in\SymmGrp}\sign{\sigma}\prod_{i=1}^n\Expect_\Gra\big[\,\eta(w_i)\xi(b_{\sigma(i)})\,\big]\,.
	$$
\end{prop}
\begin{proof}
It is a corollary of Propositions \ref{Equivalence} and \ref{WickGrass}.
\end{proof}

\newpage
\section{Symplectic fermions on $\Z^2$}
\label{sec: infinite}

\renewcommand{\gray}{{\bullet\scriptscriptstyle{\mathbf{1}}}}
\renewcommand{\white}{\circ}
\renewcommand{\black}{{\bullet\scriptscriptstyle{\mathbf{0}}}}
\newcommand{\bblack}{{\protect\tikz[baseline=-0.4ex]{\protect\draw[black,fill] (0,0) circle (1.8pt) ;}}}

This section is devoted to the study of correlation functions $\Expect_\Gra[\eta(w_1)\xi(b_1)\cdots \eta(w_n)\xi(b_n)]$ as one lets $\Gra$ grow to cover the whole plane.
By virtue of Wick's theorem ---Proposition \ref{WickProb}---,
it suffices to study the behaviour of $\Expect_\Gra[\eta(w)\xi(b)]$.
\epar

In particular, such thermodynamic limit is considered along sequences of subgraphs of the infinite square lattice that grow to cover the whole $\Z^2\subset\C$.
For that purpose, consider the following partition of $\Z^2$ --- see Figure \ref{temperleyan}:
elements with both coordinates even are called \emph{even black vertices}, elements with both coordinates odd are called \emph{odd black vertices}, and the rest are called \emph{white vertices}.
Then, let $\ZBlack$, $\ZGray$ and $\ZWhite$ denote the set of even black vertices, odd black vertices and white vertices respectively.
This way, the infinite square lattice is bipartite between black vertices $\Z^2_\bullet=\ZBlack\sqcup\ZGray$ and white vertices $\ZWhite$.

\begin{figure}[h!]
	\centering
	\begin{overpic}[scale=0.757, tics=10]{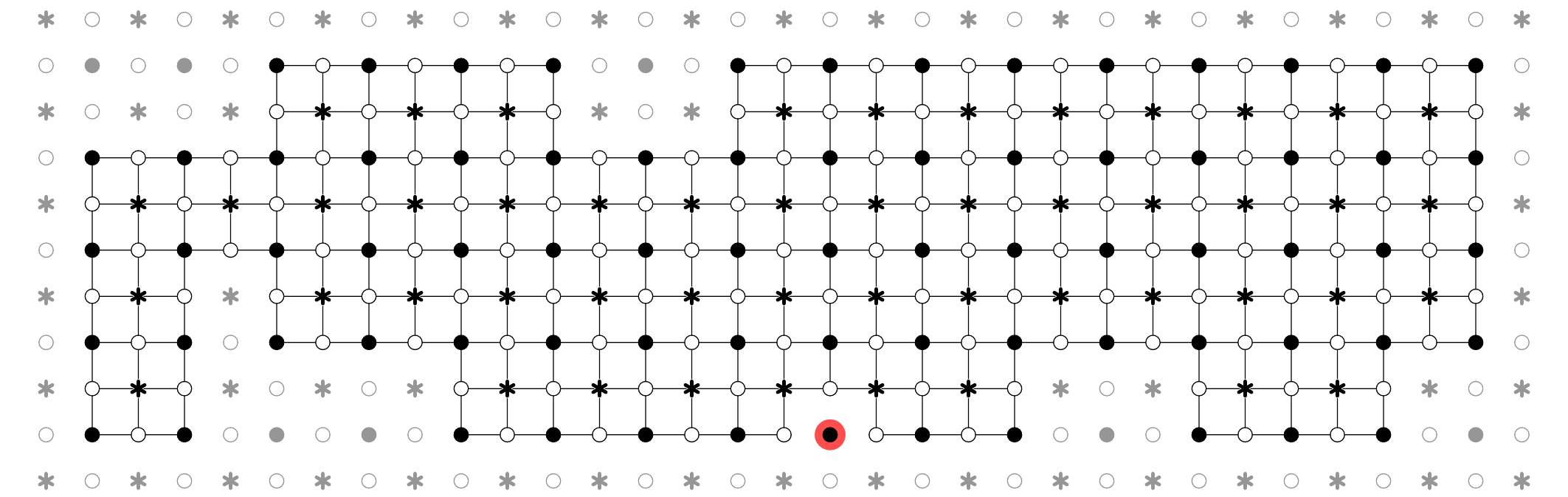}
	\end{overpic}
	\caption{{A temperleyan domain $\Ver=\Gray\sqcup\Black\sqcup\White$ with a distinguished \\ boundary point $\protectsink$ (red) and its associated dimerable graph $\Gra_{\protectsmallsink}$.}}
	\label{temperleyan}
\end{figure}

\subsection{Temperleyan domains.}\label{temper_domains}
A subset of vertices $\Ver\subset\Z^2$ is said be \emph{connected}, 
if every two vertices in $\Ver$ can be connected by a $\Z^2$-nearest-neighbour path within $\Ver$.
A finite con\-nec\-ted subset of vertices $\Ver\subset\Z^2$ is said to be a \emph{simply-connected domain} if there exists a Jordan curve made of a concatenation of length $1$ segments between nearest neighbours in $\Ver$ such that all the vertices in $\Z^2\setminus\Ver$ and no vertex in $\Ver$ lie in its exterior.
The points where this Jordan curve is not smooth are called the \emph{corners} of $\Ver$.
Note that it follows that all the vertices of a simply-connected domain have at least $2$ nearest neighbours within the domain.
A simply-connected domain is called a \emph{temperleyan domain} if all of its corners are even black --- see Figure \ref{temperleyan}.
From the partition of $\Z^2$,
a temperleyan domain inherits a partition $\Ver=\Black\sqcup\Gray\sqcup\White$.
Note $\vert\Black\sqcup\Gray\vert=\vert\White\vert + 1$.
\epar

Let $\Ver$ be a temperleyan domain and let $\sink\in\Black$ be a distinguished ---even--- black \emph{boundary point}, i.e. a black vertex such that there exists $w\in\ZWhite\setminus\Ver^{}_{\white}$ satisfying $\Vert \sink-w\Vert=1$.
Define $\Pierced{\Ver}\coloneqq\Ver\setminus\set{\sink}$ and let $\Pierced{\Gra}=(\Pierced{\Ver},\Pierced{\Edg})$ denote the ---dimerable--- graph induced by $\Pierced{\Ver}$,
i.e. $\Pierced{\Edg}\coloneqq\set{\set{z,w}\subset\Pierced{\Ver}\,\colon \norm{z-w}=1}$.
Recall from Remark \ref{square_Kast} that, on subgraphs of $\Z^2$ such as $\Pierced{\Gra}$, the holomorphic and antiholomorphic derivatives give rise to a Kasteleyn matrix. 
Those derivatives are denoted by $\dee$ and $\deebar$ and act on functions as follows: For $f\colon\Pierced{\Ver}\longrightarrow\C$
$$
\dee f(z)\coloneqq\sum_{\underset{\set{z,w}\in\Edg_{\tinysink}}{w\in\Ver_{\tinysink}}}\frac{f(w)}{w-z}
\mspace{50mu}
\textnormal{and}
\mspace{50mu}
\deebar f(z)\coloneqq\sum_{\underset{\set{z,w}\in\Edg_{\tinysink}}{w\in\Ver_{\tinysink}}}\frac{f(w)}{\cconj{w}-\cconj{z}}\,.
$$

%\begin{rmk}\label{laplacian}
%For $f\colon\SquaPierced\longrightarrow\C$ and $b\in\Black$,
%$$
%\dee\deebar f(b)
%=
%\sum_{\underset{\vert\tilde{b}- b\vert=2}{\,\tilde{b}\in\SquaPierced}}
%f(\tilde{b})-N_{\Black}(b) f(b)\,,
%$$
%where $N_{\Black}(b)\coloneqq\vert\{\,\tilde{b}\in\Black\,\colon\vert b-\tilde{b}\vert=2\,\}\vert$ is the number of nearest neighbours of $b$ in $\Black$.
%So $\dee\deebar f$ is the laplacian of $f$ with Neumann boundary conditions almost everywhere --- note the neighbours of the vertices are taken with respect to $\Black$ ---not $\SquaPierced$---, so, around $\sink$, the boundary conditions are of Dirichlet type.
%On the other hand, for $f\colon\Romb\longrightarrow\C$ and $g\in\Romb$,
%$$
%\dee\deebar f(b)
%=
%\sum_{\underset{\vert\tilde{g}- g\vert=2}{\,\tilde{g}\in\Romb}}
%f(\tilde{b})-4f(b)\,.
%$$
%So $\dee\deebar f$ is the laplacian of $f$ with Dirichlet boundary conditions.
%\hfill$\diamond$
%\end{rmk}

\subsection{Green's functions and $2$-point function.}
For any  graph $\Gra=(\Ver,\Edg)$ and any function on its vertices $f:\Ver\longrightarrow\C$,
the \emph{laplacian} of $f$ at $v\in\Ver$ is defined as
$$
\Delta^\Gra f(v)\coloneqq\sum_{\underset{\set{u,v}\in\mathrel{\raisebox{0.6pt}{$_{\mathcal{E}}$}}}{u\in\mathrel{\raisebox{1.2pt}{$_{\mathcal{V}}$}}}}\big(f(u)-f(v)\big)\,.
$$
The \emph{$2$-point function} $\Expect_{\Gra_{\smallsink}}[\,\eta(w)\xi(z)\,]$ on $\Gra_{\smallsink}$ will be proven to be closely related to the Green's function of the laplacian in two associated graphs --- see Proposition \ref{2pt-Green}.
Let us build those --- see Figure \ref{temperleyan_2}.
\epar

\begin{figure}[b!]
	\centering
	\begin{overpic}[scale=0.757, tics=10]{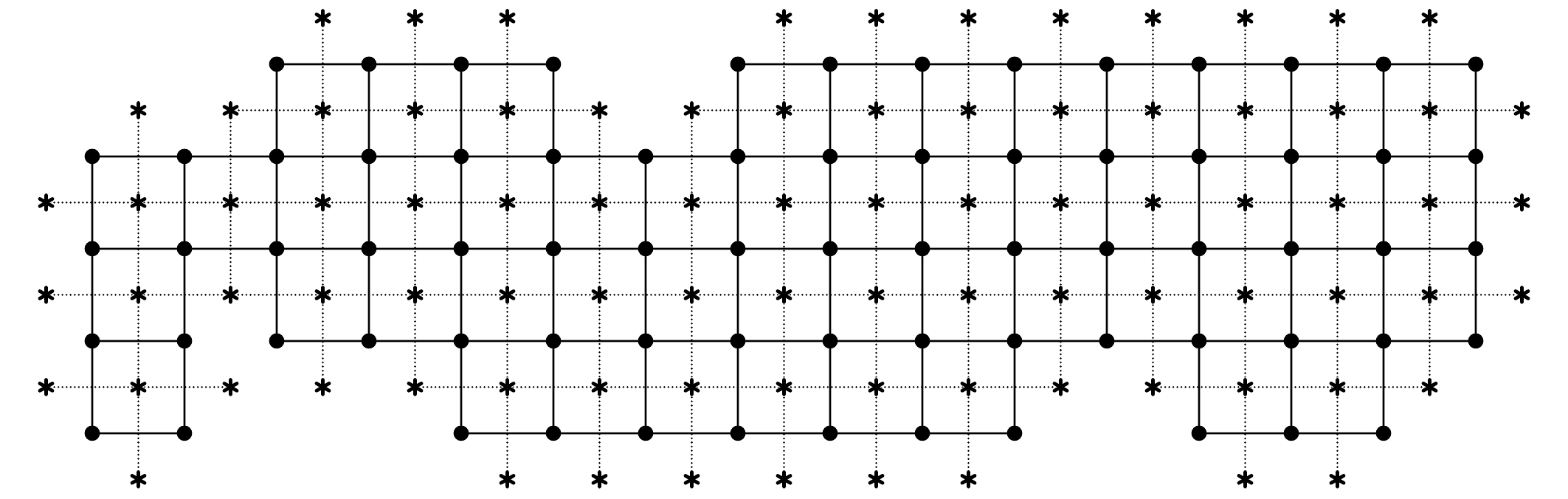}
	\end{overpic}
	\caption{The graphs $\Gra_\black$ (solid) and $\Gra_\gray$ (dotted) associated to \\ the temperleyan domain $\Ver$ in Figure \ref{temperleyan}.}
	\label{temperleyan_2}
\end{figure}

The graph $\Gra_\black=(\Black,\Edg_\black)$ is the one induced by $\Black$,
i.e. $\Edg_\black$ contains all pairs $\set{b_1,b_2}\subset\Black$ such that $\Vert b_1-b_2\Vert=2$.
Then, for $b_0\in\Black\setminus\set{\sink}$, let $\Green_\black(b_0,\,\cdot\,)$ be the Green's function on $\Gra_\black$ with Dirichlet boundary conditions at $\sink$,
i.e. the unique solution $f\colon\Black\longrightarrow\C$ to $\Delta^{\black}f(b)\coloneqq\Delta^{\Gra_\black}f(b)=-\delta_{b,b_0}$ for all $b\in\Black\setminus\set{\sink}$ with $f(\sink)=0$.
Essentially, $\Green_\black$ is the Green's function on $\Gra_\black$ with Neumann boundary conditions everywhere except at $\sink$.
\epar

As for the odd side, define the \emph{boundary of $\Gray$} as the set of odd vertices $b^*\in\ZGray\setminus\Gray$ at Manhattan distance $1$ from $\Ver$, and let it be denoted by $\bdry\Gray$.
Similarly as in the even case, let $\Gra_\gray\coloneqq(\Gray\cup\partial\Gray,\Edg_\gray^{})$ be the graph induced at distance $2$ by $\Gray$ decorated with edges from $\Gray$ to $\partial\Gray$.
For $b_0^*\in\Gray$, let $\Green_{\gray}(b_0^*,\,\cdot\,)$ be the Green's function on $\Gra_\gray$ with Dirichlet boundary conditions on $\partial\Gray$,
i.e. the unique solution $f\colon\Gray\cup\partial\Gray\longrightarrow\C$ to $\Delta^{\gray} f(b^*)\coloneqq\Delta^{\Gra_\gray}f(b^*)=-\delta_{b^*,b_0^*}$ for all $b^*\in\Gray$ and satisfying $f(b^*)= 0$ for all $b^*\in\bdry\Gray$.
\epar

Note $\Gra_\black$ and $\Gra_\gray$ are conjugate to each other in the sense that each edge of $\Gra_\black$ crosses perpendicularly an edge of $\Gra_\gray$ and vice versa.

\begin{rmk}\label{d-dbar-lap}
	Consider a function $f\colon\Pierced{\Ver}\longrightarrow\C$.
	For $b\in\Black\setminus\set{\sink}$, a simple computation yields $\dee\deebar f(b)=\Delta^\black {f}_\black(b)$ where ${f}_\black$ is the restriction of $f$ to $\Black\setminus\set{\sink}$ and extended to $\sink$ by zero.
	Similarly, for $b^*\in\Gray$, one gets $\dee\deebar f(b^*)=\Delta^\gray {f}_\gray(b^*)$ where ${f}_\gray$ is the restriction of $f$ to $\Gray$ and extended to $\bdry\Gray$ by zero.
	\hfill$\diamond$
\end{rmk}

\begin{prop}\label{2pt-Green}
Let $w\in\White$ be a white vertex not adjacent to $\sink$.
Let $b^*_1,b^*_2\in\Gray$ and $b_1,b_2\in\Black$ be the four $\Z^2$-nearest neighbours of $w$.
Then, for $z\in\Gray\cup\Black\setminus\set{\sink}$,
$$
\frac{1}{2}\,
\Expect_{\Gra_{\Ver}^\times}[\,\eta(w)\xi(z)\,]\,
=
\begin{cases}
\mspace{10mu}
\frac{1}{b_1-b_2}\,
\Big(
\Green_\black(b_1,z)-\Green_\black(b_2,z)
\Big)
\mspace{20mu}
&
\textnormal{if}
\mspace{10mu}
z\in\Black\setminus\set\sink
\\
\ 
\\
\mspace{10mu}
\frac{1}{b^*_1-b^*_2}\,
\Big(
\Green_\gray(b^*_1,z)-\Green_\gray(b^*_2,z)
\Big)
\mspace{20mu}
&
\textnormal{if}
\mspace{10mu}
z\in\Gray
\end{cases}\,.
$$
\end{prop}
\begin{proof}
Let $f_\gray$ be the restriction of $z\mapsto\Expect_{\Gra_{\tinysink}}[\eta(w)\xi(z)]$ to $\Gray$ extended to $\bdry\Gray$ by zero.
Theorem \ref{discrete_holomorphicity} states $\deebar_{\tilde{w}}\Expect_{\Gra_{\tinysink}}[\eta(w)\xi(\tilde{w})]=\delta_{w,\tilde{w}}$ for $\tilde{w},w\in\White$.
Hence, by Remark \ref{d-dbar-lap}, for any $b^*\in\Gray$,
$$
\Delta^\gray f_\gray(b^*)
=
(\dee\deebar)_{b^*}\Expect_{\Gra_{\tinysink}}[\eta(w)\xi(b^*)]
=
\sum_{\underset{\set{b^*,\tilde w}\in\Edg_{\tinysink}}{\tilde w\in\Ver_{\tinysink}}}\frac{\deebar_{\tilde w}\Expect_{\Gra_{\tinysink}}\big[\eta(w)\xi(\tilde w)\big]}{\tilde w-b^*}
=
\frac{\delta_{b^*_1,b^*}}{w-b^*_1}
+
\frac{\delta_{b^*_2,b^*}}{w-b^*_2}\,.
$$
Moreover, $f_\gray$ satisfies the right boundary conditions, i.e. $f_\gray(b^*)=0$ for $b^*\in\bdry\Gray$, and so,
the claim follows from the uniqueness of the Green's function and $w-b^*_1=b^*_2-w$.
Since $w$ is not adjacent to $\sink$, the proof for even black vertices is identical.
\end{proof}

\begin{figure}[h!]
	\centering
	\begin{overpic}[scale=0.75, tics=10]{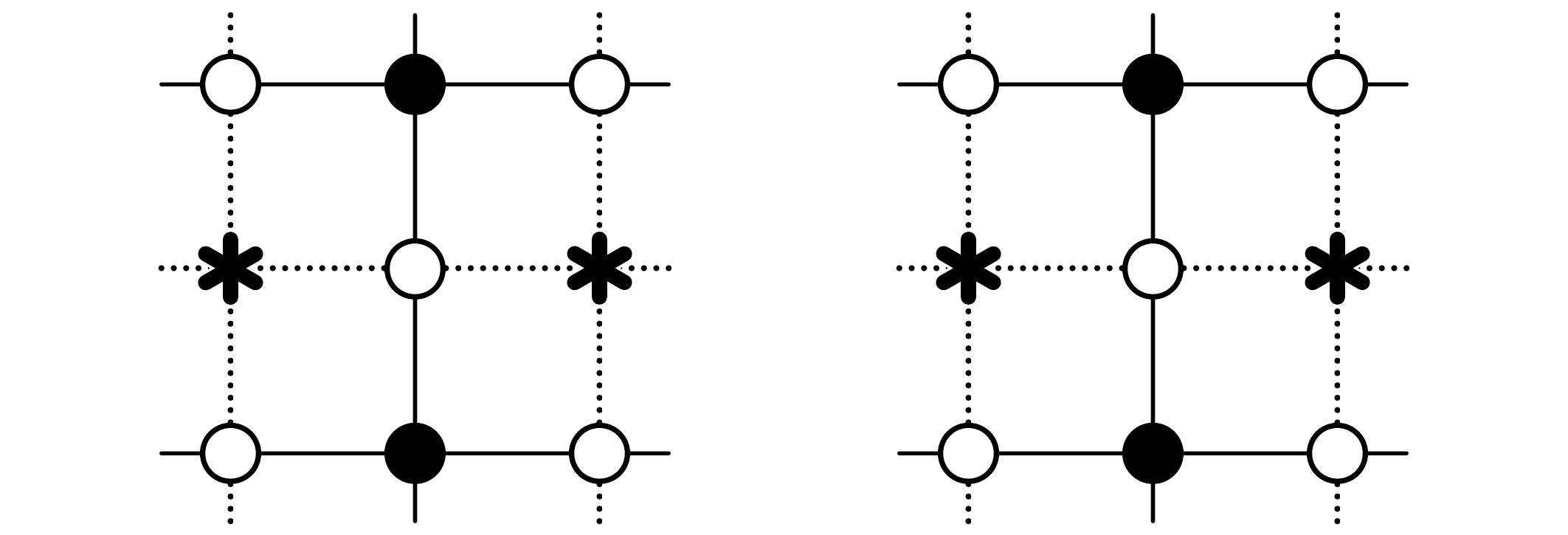}
		\put(34,20){\large $b^*_R$}
		\put(10.5,20){\large $b^*_L$}
		\put(22.2,31.7){\large $b_+$}
		\put(22.2,8){\large $b_-$}
		\put(22.9,19.5){\large $\tilde{w}$}
		\put(82,20){\large $b^*_1$}
		\put(58.5,20){\large $b^*_2$}
		\put(70.2,31.7){\large $b_1$}
		\put(70.2,8){\large $b_2$}
		\put(70.2,19.5){\large ${w}$}
	\end{overpic}
	\caption{\centering The four $\Z^2$-nearest neighbours of a white vertex.}
	\label{neighbours}
\end{figure}

\begin{rmk}\label{harm_conj-finite}
Take the vertices $b^*_1,b^*_2,b_1,b_2$ in Proposition \ref{2pt-Green} to satisfy the relation $b^*_1-b^*_2=-\ii(b_1-b_2)$ --- see Figure \ref{neighbours}.
The functions $f_\black(b)\coloneqq \Green_\black(b_1,b)-\Green_\black(b_2,b)$
and $f_\gray(b^*)\coloneqq \Green_\gray(b^*_1,b^*)-\Green_\gray(b^*_2,b^*)$ defined on $\Black$ and $\Gray \cup\bdry\Gray$,
respectively,
are \emph{harmonic conjugates} of each other on $\White\setminus\set{w}$ in the following sense:
For any $\tilde{w}\in\White\setminus\set{w}$, let $b_+,b_-\in\Black$ be the $\Z^2$-nearest neighbours of $\tilde{w}$,
and let $b^*_L,b^*_R\in\Gray\cup\bdry\Gray$ be the $\Z^2$-nearest neighbours of $\tilde{w}$ sitting on the left and right, respectively, when going from $b_-$ to $b_+$ through $\tilde{w}$ --- see Figure \ref{neighbours}.
Then,
$$
\mspace{200mu}
f_\black(b_+)-f_\black(b_-)
=
f_\gray(b^*_R)-f_\gray(b^*_L)\,.
\mspace{200mu}
\diamond
$$
\end{rmk}

\subsection{Thermodynamic limit.}
From Proposition \ref{2pt-Green}, one expects the $2$-point function to converge to the derivative of the full-plane Green's function on the two sublattices $\ZBlack$ and $\ZGray$ as one lets the temperleyan domain $\Ver$ grow to cover the whole $\Z^2$.
The full-plane Green's func\-tion $\Green$ is the unique function on $\ZBlack$ satisfying $\Green(0)=0$, $\Delta \Green(z)=-\delta_{z,0}$ with asymptotic behaviour
$$
\Green(z) = - \frac{1}{2\pi}\log{\vert z\vert}+C+O\bigg(\frac{1}{\vert z\vert^{2}}\bigg)
$$
as $\vert z\vert\to\infty$, where $C=-(\gamma+\frac{3}{2}\log 2)/2\pi$ and $\gamma$ is Euler's constant \cite{LawLim}.
\epar

\begin{rmk}\label{harm_conj-infinite}
Similarly as in Remark \ref{harm_conj-finite}, one can build two functions in terms of the full-plane Green's function that are harmonic conjugates of each other on $\ZWhite$ except for one point.
Fix a white vertex $w\in\ZWhite$, and let $b_1,b_2\in\ZBlack$ and $b^*_1,b^*_2\in\ZGray$ be its four $\Z^2$-nearest neighbours satisfying $b^*_1-b^*_2=-\ii(b_1-b_2)$ as in Remark \ref{harm_conj-finite} --- see Figure \ref{neighbours}.
Then the functions $F_\black(b)\coloneqq \Green(b_1-b)-\Green(b_2-b)$
and $F_\gray(b^*)\coloneqq \Green(b^*_1-b^*)-\Green(b^*_2-b^*)$ defined, respectively, on $\ZBlack$ and $\ZGray$ are harmonic conjugate of each other on $\ZWhite\setminus\set{w}$ in the same sense as in Remark \ref{harm_conj-finite}
\hfill$\diamond$
\end{rmk}

For the rest of this section, fix a sequence $(\Ver^{n})_{n\in\N}$ of temperleyan domains with distinguished black boundary points $\sink_n\in\Black^n$ that {converges to $\Z^2$},
i.e. for every $k\in\N$ there exists $N\in\N$ such that $\Ball(0;k)\subset\Ver^n$ for all $n\geq N$,
where $\Ball(z;r)\coloneqq\set{w\in\Z^2\,\colon \norm{z-w}<r}$.
Let $\Ver^n\uparrow\Z^2$ denote such convergence.
For such a se\-quence, let $\Gra_{\smallsink}^n$, $\Gra_\black^n$, $\Gra_\gray^n$, $\Green_{\black}^{(n)}$ and $\Green_{\gray}^{(n)}$ denote the objects described in the previous subsections for the temperleyan domain $\Ver^n$.

\begin{thm}\label{the_convergence}
For any $w\in\ZWhite$ and $z\in\ZBlack\sqcup\ZGray$,
$$
\frac{1}{2}\,
\Expect_{\Gra_{\tinysink}^n}\big[\,\eta(w)\xi(z)\,\big]
\ \ \xrightarrow{\ \ \ n\to\infty\ \ \ }\ \ \ 
\frac{\Green(z-w_1)-\Green(z-w_2)}{w_1-w_2}\,
$$
where $w_1,w_2\in\ZBlack\sqcup\ZGray$ are the two $\Z^2$-nearest neighbours of $w$ of the same parity as $z$.
\end{thm}

As a straight-forward corollary, one gets the existence of all correlation functions of discrete symplectic fermions in the thermodynamic limit along temperleyan domains.

\begin{coro}
Fix $w_1,\ldots,w_k\in\ZWhite$ and $z_1,\ldots,z_k\in\Z^2_\bullet$.
$$
\lim_{n\to\infty}
\Expect_{\Gra_{\tinysink}^n}\big[\,\eta(w_1)\xi(z_1)\cdots\eta(w_k)\xi(z_k)\,\big]
=
\sum_{\sigma\in{\SymmGrp_n}}
\sign\sigma\prod_{i=1}^k
\bigg(
\lim_{n\to\infty}
\Expect_{\Gra_{\tinysink}^n}\big[\,\eta(w_i)\xi(z_{\sigma(i)})\,\big]
\bigg)\,.
$$
\end{coro}
\begin{proof}
It is a consequence of Wick's theorem ---Proposition \ref{WickProb}--- Theorem \ref{the_convergence}
\end{proof}

Moreover, note that, if one fixes $w\in\ZWhite$, the asymptotic behaviour of $\lim_{n\to\infty}\Expect_{\Gra_{\tinysink}^n}[\eta(w)\xi(z)]$ as $\vert z\vert\to\infty$ is proportional to $\vert z-w\vert^{-1}+O(\vert z-w\vert^{-2})$,
which justifies the discretisation of symplectic fermions as a suitable one \cite{Kausch2}.

\subsection{Proof of Theorem \ref{the_convergence}.}
Because of the different boundary conditions in each of the graphs $\Gra_\black^n$ and $\Gra_\gray^n$, the proof of Theorem \ref{the_convergence} does not follow the same lines in both cases.
Let us treat first the simpler case: The odd half, in which the boundary conditions are Dirichlet everywhere on the boundary of $\Gra_\gray^n$.

\begin{prop}\label{converg-dirichlet}
Fix two odd black vertices $b^*_1,b^*_2\in\ZGray$ satisfying $\Vert b^*_1-b^*_2\Vert=2$.
Then, for all $b^*\in\ZGray$,
$$
\Green_\gray^{(n)}(b^*_1,b^*)-\Green_\gray^{(n)}(b^*_2,b^*)
\ \ \xrightarrow{\ \ \ n\to\infty\ \ \ }\ \ 
\Green(b^*-b^*_1)-\Green(b^*-b^*_2).
$$
\end{prop}
\begin{proof}
Consider the function $h_\gray^{(n)}$
on $\Gray^n\cup\bdry\Gray^n$ given by $h_\gray^{(n)}\coloneqq f^{(n)}_\gray-F_\gray$ with $f_\gray^{(n)}$ as in Remark \ref{harm_conj-finite} and $F_\gray$ as in Remark \ref{harm_conj-infinite}.
It satisfies $\Delta^\gray h_\gray^{(n)}(b^*)=0$ for $b^*\notin\bdry\Gray^n$ and $h_\gray^{(n)}(b^*)=-\Green(b^*-b^*_1)+\Green(b^*-b^*_2)$ for all $b^*\in\partial\Gray^n$.
The maximum principle dictates, for any $b^*\in\Gray$,
$$
\vert h_\gray^{(n)}(b^*)\vert
\leq
\max_{b^*\in\bdry\Ver_\gray^n}\vert h_\gray^{(n)}(b^*)\vert
=
\max_{b^*\in\bdry\Ver_\gray^n}\vert \Green(b^*-b^*_1)-\Green(b^*-b^*_2)\vert\,,
$$
and the asymptotic behaviour of $\Green$ makes the function $b^*\mapsto \Green(b^*-b^*_1)-\Green(b^*-b^*_2)$ be $O(1/\vert b^*\vert)$ as $\vert b^*\vert \to\infty$.
Then, the property $\Ver^n\uparrow\Z^2$ ensures $\vert h_\gray^{(n)}(b^*)\vert\to 0$ as $n\to\infty$ since the boundary $\partial\Gray^n$ only gets farther from $b^*_1$ and $b^*_2$ as $n$ grows.
\end{proof}

The same statement for the graphs $\Gra_\black^n$ is slightly more intricate to prove as a conse\-quence of the boundary conditions, which are Neumann everywhere except at $\sink_n$, where they are Dirichlet.
The proof is simple if one takes a clever choice of distinguished points.
For a distinguished point $\tilde{\sink}_n$ possibly different from ${\sink}_n$, let $\tilde\Gra_{\smallsink}^n$ and $\tilde\Green_{\black}^{(n)}$ denote the objects described in the previous subsections.

\begin{prop}\label{Neumann-close_sink}
	Fix two even black vertices $b_1,b_2\in\ZBlack$ satisfying $\Vert b_1-b_2\Vert=2$.
	For $b\in\ZBlack$, let $\tilde{\sink}_n\in\Black^n$ be any of the black boundary vertices that is closest to $b$.
	Then,
	$$
	\tilde\Green_\black^{(n)}(b_1,b)-\tilde\Green_\black^{(n)}(b_2,b)
	\ \ \xrightarrow{\ \ \ n\to\infty\ \ \ }\ \ \ 
	\Green(b-b_1)-\Green(b-b_2).
	$$
\end{prop}
\begin{proof}
	Let the vertices $b^*_1,b^*_2\in\ZGray$ and the functions $f_\gray^{(n)}$ and $\tilde{f}_\black^{(n)}$ be as in Remark \ref{harm_conj-finite}, and let the functions $F_\black$ and $F_\gray$ be as in Remark \ref{harm_conj-infinite}.
	Consider the functions $h_\gray^{(n)}$ and $h_\black^{(n)}$ on $\Ver^n_\gray$ and $\Ver^n_\black$ given by $h_\black^{(n)}\coloneqq f_\black^{(n)}-F_\black$ and $h_\gray^{(n)}\coloneqq f_\gray^{(n)}-F_\gray$, respectively.
	Note they are harmonic conjugates of each other except at the edges $\set{b^*_1,b^*_2}$ and $\set{b_1,b_2}$.
	Note, too, that they are harmonic on their domains.
	Thus, by Harnack's estimate, there exists a constant $C>0$ such that
	$$
	\big\vert h_\gray^{(n)}(\tilde b^*_1)-h_\gray^{(n)}(\tilde b^*_2)\big\vert
	\leq
	C\,\frac{\max_{b^*\in\bdry\Ver^n_\gray}\vert h_\gray^{(n)}(b^*)\vert}{\min_{b^*\in\bdry\Ver^n_\gray}\Vert \tilde b^*_1-b^*\Vert}
	=
	C\,\frac{\max_{b^*\in\bdry\Ver^n_\gray}\vert F_\gray(b^*)\vert}{\min_{b^*\in\bdry\Ver^n_\gray}\Vert \tilde b^*_1-b^*\Vert}
	$$
	for any $\tilde b^*_1,\tilde b^*_2\in\Ver^{n}_\gray$ satisfying $\Vert \tilde b^*_1-\tilde b^*_2\Vert=2$.
	Consider any of the shortest $\ZBlack$-nearest-neighbour path $\gamma_n\coloneqq(b_0,b_1,\ldots,b_k)$ in $\Ver^n_\black$ from $b_0=b$ to $b_k=\tilde\sink_n$ --- see Figure \ref{neighbours}. 
	Note $k=\Vert \tilde\sink_n-b\Vert$.
	Let $b^*_{L,i},b^*_{R,i}\in\ZGray$ be the odd black vertices respectively to the left and right when going from $b_{i-1}$ to $b_i$.
	Then,	
	\begin{align*}
	\big\vert h^{(n)}_\black(\tilde\sink_n)\, - & \,h^{(n)}_\black(b)\big\vert =\,
	\Bigg\vert
	\sum_{i=1}^{\Vert \tilde\smallsink_n-b\Vert}
	\Big(h^{(n)}_\black(b_i)-h^{(n)}_\black(b_{i-1})\Big)
	\Bigg\vert
	=
	\Bigg\vert
	\sum_{i=1}^{\Vert \tilde\smallsink_n-b\Vert}
	\Big(h^{(n)}_\gray(b^*_{R,i})-h^{(n)}_\gray(b^*_{L,i})\Big)
	\Bigg\vert
	\\
	\leq \,&
	\sum_{i=1}^{\Vert \tilde\smallsink_n-b\Vert}
	\Big\vert h^{(n)}_\gray(b^*_{R,i})-h^{(n)}_\gray(b^*_{L,i})\Big\vert
	\leq
	C\sum_{i=1}^{\Vert \tilde\smallsink_n-b\Vert}
	\frac{\max_{b^*\in\bdry\Ver^n_\gray}\vert F_\gray(b^*)\vert}{\Vert \tilde\sink_n-b\Vert-i+1}
	\ \ \xrightarrow{\ \ \ n\to\infty\ \ \ }\ \ \ 0
	\end{align*}
since, for large $n$, the sum of inverses grows like $\log(\Vert \tilde\sink_n\Vert)$ as $n\to\infty$ but $F_\gray$ evaluated on the boundary of $\Ver^n$ goes to $0$ as $1/\Vert \tilde\sink_n\Vert$ by the asymptotic behaviour of $\Green$.
Moreover, $h^{(n)}_\black(\tilde\sink_n)\longrightarrow 0$ as $n\to\infty$ again by the asymptotic behaviour of $\Green$, which completes the proof.
\end{proof}

\begin{figure}[t!]
	\centering
	\begin{overpic}[scale=0.75, tics=10]{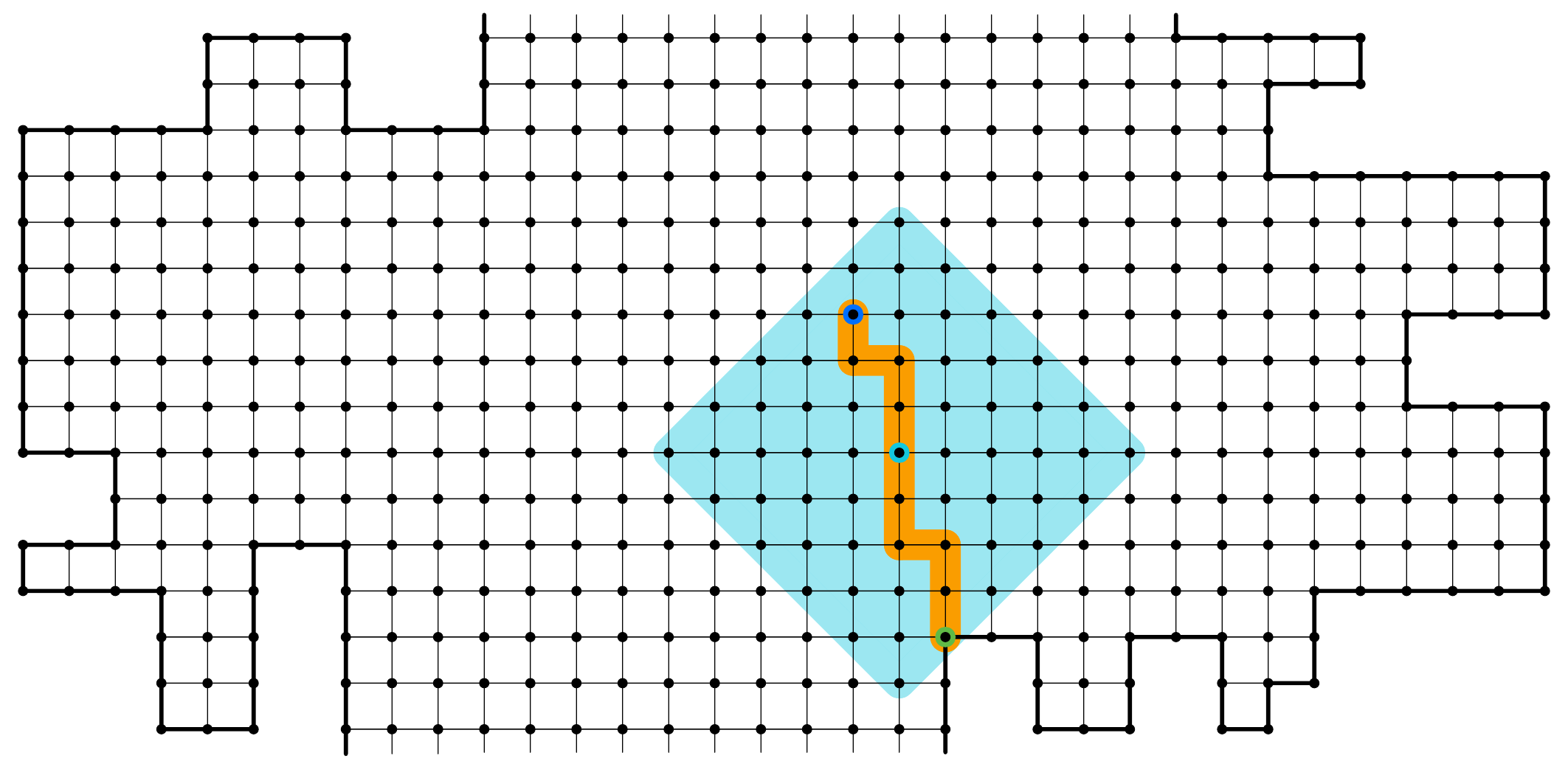}
	\end{overpic}
	\caption{A path $\gamma_n$ (orange) from $b\in\Ver_\black^n$ (dark blue) to the closest \\ boundary point $\tilde\protectsink_n\in\Ver_\black^n$ (green), and a point $b_i\in\Black^n$ (light blue) \\ and the largest ball centered thereat contained in $\Black^n$.}
	\label{distinguished}
\end{figure}

So as to prove the same statement for arbitrary choices of distinguished boundary points $\sink_n$, one needs to first prove an auxiliary result about dimers on temperleyan domains.
In a dimer configuration $\omega\in\Dim{\Gra}$, say that an edge $e$ is \emph{open} if $e\in\omega$.

\begin{coro}
	Fix two $\Z^2$-nearest neighbours $w\in\ZWhite$ and $b\in\ZBlack$.
	Then
	$$
	\Prob_{\tilde\Gra_{\tinysink}^n}\big[\set{w,b}\ \textnormal{open}\,\big]
	\ \ \xrightarrow{\ \ \ n\to\infty\ \ \ }\ \ \ 
	\frac{1}{4}\,.\mspace{80mu}
	$$
\end{coro}
\begin{proof}
	It follows from Corollary \ref{probabilistic} and the values of the full-plane Green's function in a neighbourhood of the origin: $\Green(0)=0$ and $\Green(z)=1/4$ for all $z\in\ZBlack$ with $\Vert z\Vert =2$.
\end{proof}

Note the following lemma differs from the previous corollary by just one tilde, i.e. in the following result one allows for any choice of distinguished points $\sink_n\in\Black^n$.

\begin{lemma}\label{coupling}
	Fix two $\Z^2$-nearest neighbours $w\in\ZWhite$ and $b\in\ZBlack$.
	Then
	$$
	\Prob_{\Gra_{\tinysink}^n}\big[\set{w,b}\ \textnormal{open}\,\big]
	\ \ \xrightarrow{\ \ \ n\to\infty\ \ \ }\ \ \ 
	\frac{1}{4}\,.\mspace{80mu}
	$$
\end{lemma}

\begin{figure}[t!]
	\centering
	\begin{overpic}[scale=0.75, tics=10]{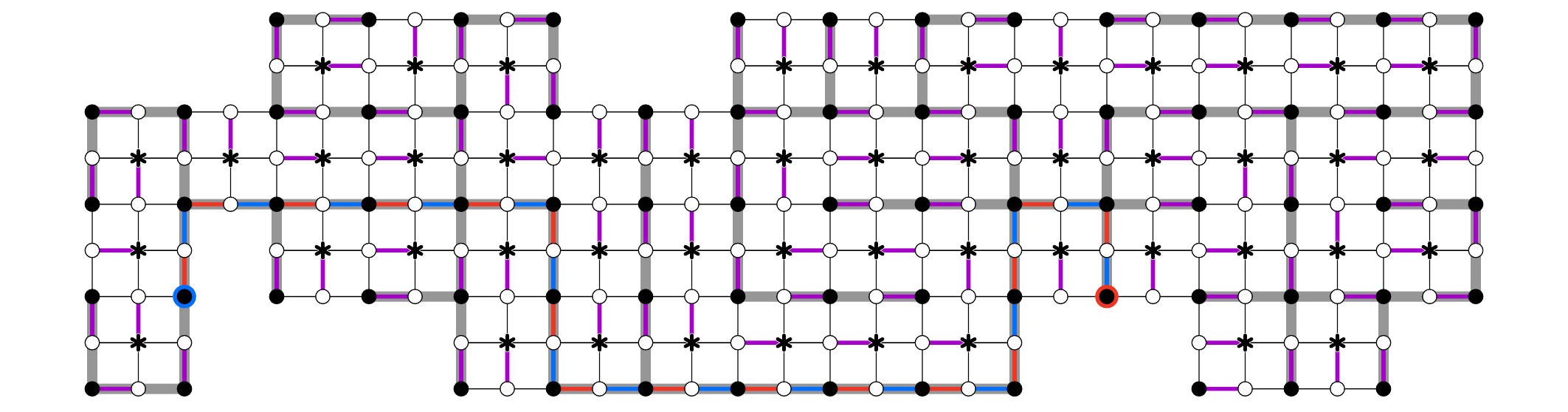}
		\put(13.5,5){$\tilde\sink_n$}
		\put(68,4.2){$\sink_n$}
	\end{overpic}
	\caption{A tree on $\Gra_\black^n$ and the dimer configurations on $\Gra_{\protectsmallsink}^n$ and $\tilde\Gra_{\protectsmallsink}^n$ associated to it through the Temperley bijection.
		In purple the dimers that are open both in $\Gra_{\protectsmallsink}^n$ and $\tilde\Gra_{\protectsmallsink}^n$.}
	\label{fig-coupling}
\end{figure}

\begin{proof}
The Temperley bijection provides couplings between uniform spanning trees on $\Gra_\black^n$, and uniform dimers on $\Gra_{\smallsink}^n$ and $\tilde\Gra_{\smallsink}^n$ as shown in Figure \ref{fig-coupling}.
Let $\tau_n$ and $\tilde\tau_n$ be the bijections between the set $\mathcal{T}_\black^n$ of spanning trees on $\Gra_\black^n$, and $\Dim{\Gra_{\smallsink}^n}$ and $\Dim{\tilde\Gra_{\smallsink}^n}$ respectively.
In turn, those couplings provide a coupling between $\Dim{\Gra_{\smallsink}^n}$ and $\Dim{\tilde\Gra_{\smallsink}^n}$ as follows:
Let $b,b'\in\ZBlack$ and $b^*_1,b^*_2\in\ZGray$ be the $\Z^2$-nearest neighbours of $w$.
Consider the event $A_n\subset\mathcal{T}_\black^n$ that the branch from $\sink_n$ to $\tilde\sink_n$ contains the edge $\{b,b'\}$.
For a tree $T\notin A_n$, the dimer $\set{w,b}$ has the same state simultaneously ---open or closed--- in $\tau_n(T)$ and $\tilde\tau_n(T)$ --- see Figure \ref{fig-coupling}.
Then, if one proves the probability of $A_n$ to converge to $0$ the proof is complete.
This can be accomplished using Wilson's algorithm to generate uniform spanning trees and the Beurling estimate:
Consider the coupling (bijection) $\phi$ between $\mathcal{T}_\black^n$ and the set $\mathcal{T}_\gray^n$ of spanning trees on $\Gra_\gray^n$ wired at the boundary $\bdry\Ver_\gray^n$ --- see Figure \ref{Beurling}. Then, $\phi(A_n)$ is contained in the event $B_n\subset\mathcal{T}_\gray^n$ that the branch of $b^*_1$ and the branch of $b^*_2$ do not overlap --- see Figure \ref{Beurling}.
Using Wilson's algorithm, the probability of $B_n$ can be bounded from above by the probability of the event $C_n$ that a random walk on $\Gra_\gray^n$ started at $b^*_2$ hits the boundary $\bdry\Ver_\gray^n$ before hitting a branch connecting $b^*_1$ to $\bdry\Ver_\gray^n$.
The Beurling estimate asserts there exists a constant $C>0$ such that
$$
\Prob[C_n]\leq C\,\bigg(\frac{1}{\min_{g\in\bdry\Ver^n_\gray}\Vert {g}_2-g\Vert}\bigg)^{\frac{1}{2}}\,,
$$
which converges to $0$ as $n\to\infty$ by $\Ver^n\uparrow\Z^2$.
Tracing back the couplings, this implies $\Prob[A_n]\longrightarrow 0$ as $n\to\infty$.
\end{proof}

\begin{figure}[h!]
	\centering
	\begin{overpic}[scale=0.75, tics=10]{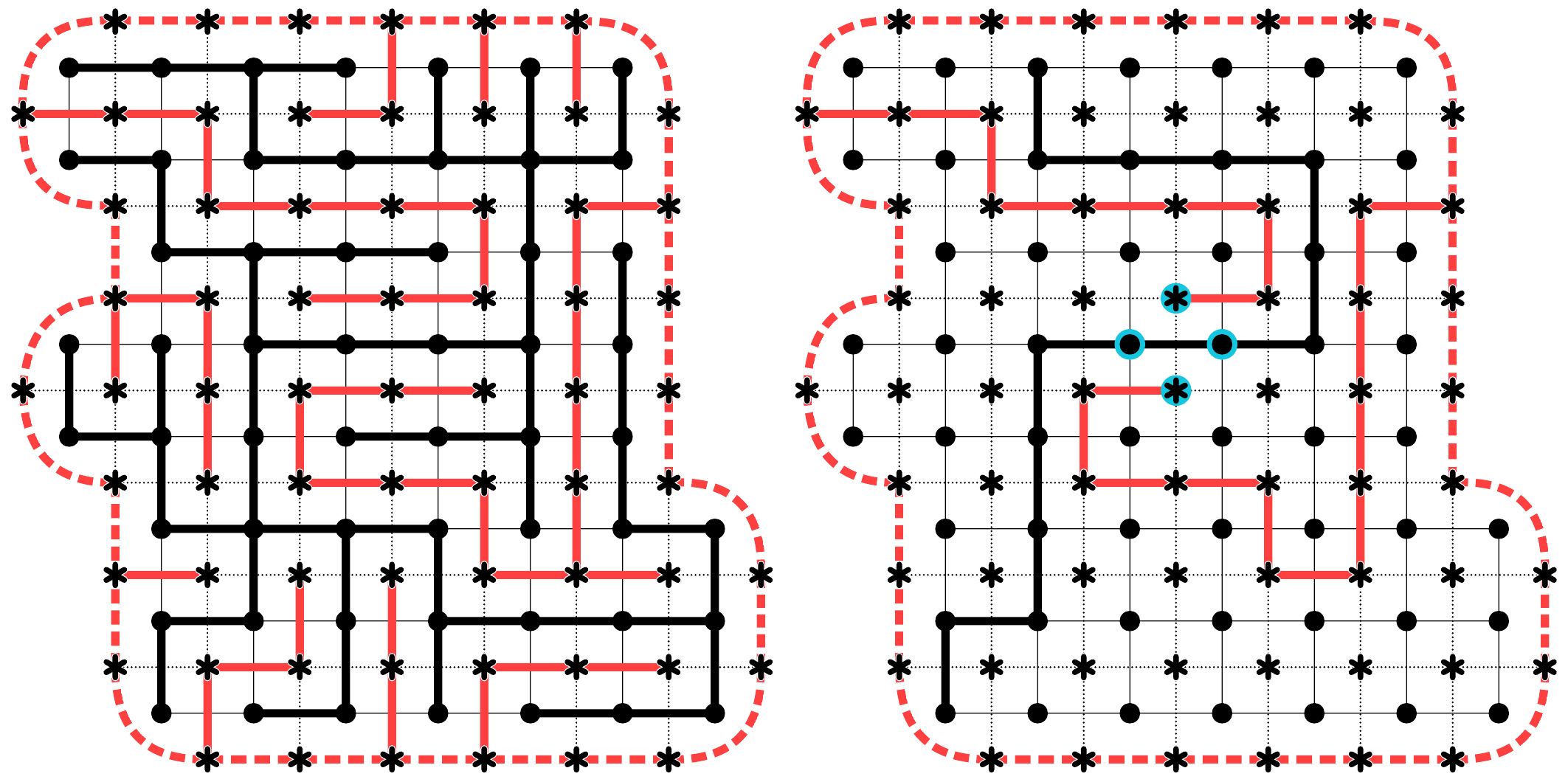}
	\end{overpic}
	\caption{Left: The coupling between $\mathcal{T}_{\black}^n$ and $\mathcal{T}_{\gray}^n$.
		Right: An instance of $\phi(A_n)\subset B_n$. 
		\\
		The vertices $b,b'\in\Black^n$ and $b^*_1,b^*_2\in\Ver^n_\gray$ highlighted in blue.}
	\label{Beurling}
\end{figure}

Note, again, that the following proposition differs from Proposition \ref{Neumann-close_sink} by a few tildes.

\begin{prop}
	Fix two even black vertices $b_1,b_2\in\ZBlack$ satisfying $\Vert b_1-b_2\Vert=2$.
	Then, for any $b\in\ZBlack$,
	$$
	\Green_\black^{(n)}(b_1,b)-\Green_\black^{(n)}(b_2,b)
	\ \ \xrightarrow{\ \ \ n\to\infty\ \ \ }\ \ \ 
	\Green(b-b_1)-\Green(b-b_2).
	$$
\end{prop}
\begin{proof}
The arguments have the same flavour as in the proof of Proposition \ref{Neumann-close_sink}, but now taking a fixed path from $b^*$ to $b^*_1$ and using Corollary \ref{probabilistic}, Lemma \ref{coupling} and the values of the full-plane Green's function $\Green(0)=0$ and $\Green(z)=-1/4$ for $z=\pm1,\pm\ii$.
\end{proof}

\newpage
\section{Local fields and current modes}
\label{sec: fields-currents}

\renewcommand{\black}{\bullet}
\renewcommand{\white}{\circ}

In what follows, the attention is brought onto discrete symplectic fermions on arbitrary domains of $\Z^2$ with no holes --- see Section \ref{sec: infinite}.
For the rest of the text, the Kasteleyn matrix on any such domain is fixed to be $\partial$ --- see Remark \ref{square_Kast}.
\epar

In particular, in this section, the construction of the space of local fields of symplectic fermions on the square lattice is presented.
That one is the space that later on will be proven to carry a representation of the Virasoro algebra.
\epar

In the continuum \cite{Kausch2}, the fermions $\eta$ and $\xi$ can be put in an equal footing by defining the two-component fermion field $\chi=(\chi^+,\chi^-)$, where $\chi^+(z)=\dee\xi(z)$ and $\chi^-(z)=\eta(z)$.
The algebraic content of symplectic fermions in the Vertex Operator Algebra (VOA) sense is encoded in the so called \emph{current modes} of $\chi$
i.e. the coefficients of the formal series $\chi^\alpha(z) = \sum_{k\in\Z}\,\chi^\alpha_k\,z^{-k-1}$.
Although no such construction is intended to be translated to the discrete, there are two tools of discrete complex analysis on $\Z^2$ {introduced in \cite{HKV}} that allow one to define such current modes:
a bilinear notion of discrete integration,
and a family of functions that mimic the properties of the complex Laurent monomials $\C\ni z\mapsto z^n$ for $n\in\Z$.
Let us review them here.

For the rest of the section, let $\Vert\cdot\Vert$ denote the Manhattan norm, i.e. $\Vert z\Vert=\vert\re z\vert+\vert\im z\vert$,
and let $\Ball(x;r)$ denote the ball of radius $r\geq0$ centered at $x\in\Z^2$ with respect to the Manhattan distance.
\vspace{-10pt}

\subsection{Preliminaries: Discrete integration and discrete monomials.}
Discrete integration is performed along dual contours.
A \emph{dual contour} $\gamma=(p_0,\ldots,p_n)$ with $p_0=p_n$ is a sequence of consecutively nearest \textit{plaquette centers} that does not intersect itself,
i.e. $p_i\in(\Z^2)^*=(\Z+\frac{1}{2})^2\subset\C$ for $i=0,\ldots, n$ and $|p_i-p_{i-1}|=1$ for $i=1,\ldots,n$ and such that $p_1,\ldots,p_n$ are all distinct.
Then, $\interior_\black\gamma\subset\ZBlack$ and $\interior_\white\gamma\subset\ZWhite$ denote, respectively, the set of black and white vertices enclosed by $\gamma$ ---see Figure \ref{contour}---, and $\interior{\gamma}\coloneqq\interior_\black{\gamma}\cup\interior_\white{\gamma}$.
A dual contour is said to be \emph{positively oriented} if it turns counterclockwise around its interior.
\epar

\begin{figure}[h!]
	\centering
	\begin{overpic}[scale=0.75, tics=10]{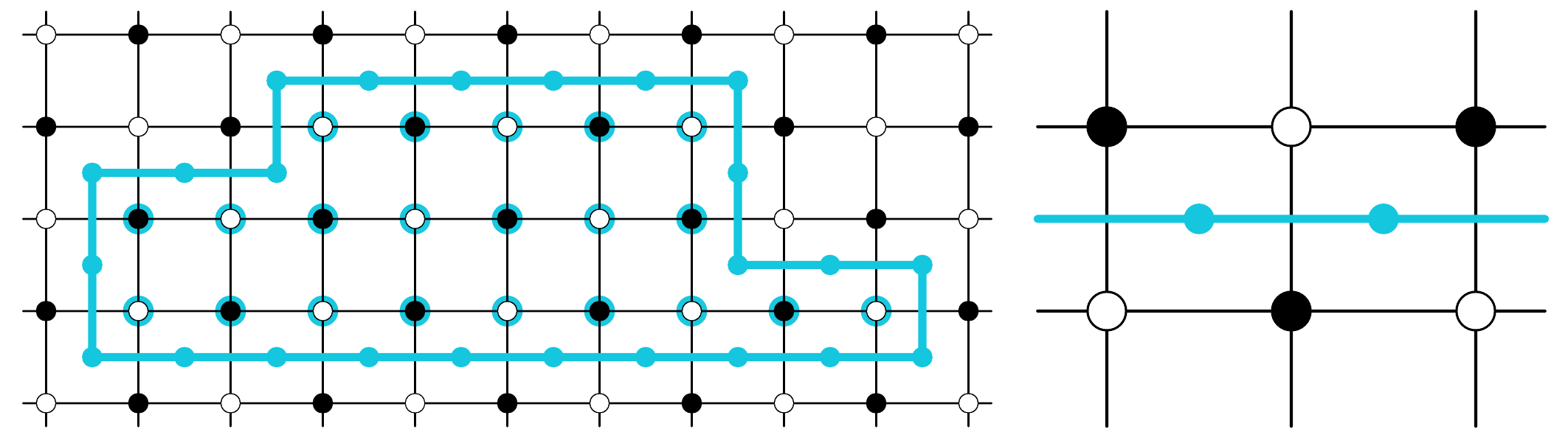}
		\put(83.5,21.2){$w_k$}
		\put(79.3,4.8){$b_k$}
		\put(89,11){$p_k$}
		\put(73.5,11){$p_{k-1}$}
	\end{overpic}
	\caption{A dual contour and the vertices in its interior highlighted in blue.}
	\label{contour}
\end{figure}
\newpage

For any vector space $\textnormal{V}$ over $\C$ and any pair of functions $f:\ZWhite\longrightarrow\C$ and $g:\ZBlack\longrightarrow \textnormal{V}$, one defines the integral 
$$
\sqint_\gamma f(z_\white)g(z_\black)\dd{z}
\coloneqq
\sum_{k=1}^n (p_k-p_{k-1})f(w_k)g(b_k)\,,
$$
where $w_k\in\ZWhite$ and $b_k\in\ZBlack$ are the white and black vertices across the dual edge $\{p_k,p_{k-1}\}$ with respect to each other --- see Figure \ref{contour}.
The definition is identical when $f$ is $\textnormal{V}$-valued and $g$ is $\C$-valued.
\epar

This notion of integration is closely related to the notion of discrete differentiation given by the derivatives
\begin{align}\label{derivatives}
\dee f(z)\coloneqq\sum_{\underset{\vert z-w\vert =1}{w\in\Z^2}}\frac{f(w)}{w-z}
\mspace{50mu}
\textnormal{and}
\mspace{50mu}
\deebar f(z)\coloneqq\sum_{\underset{\vert z-w\vert =1}{w\in\Z^2}}\frac{f(w)}{\cconj{w}-\cconj{z}}
\end{align}
through the \emph{discrete Stoke's formula}:
Let $\gamma$ be a positively-oriented dual contour, then
$$
\sqint_\gamma f(z_\white)g(z_\black)\dd{z}=\ii \sum_{b\in\interior_\black\gamma}\deebar f(b)g(b)+\ii \sum_{w\in\interior_\white\gamma}f(w)\deebar g(w)\,.
$$

A function $f$ is said to be \emph{discrete holomorphic} at $z\in\Z^2$ if $\deebar f(z)=0$. 
Then, Stokes' formula implies that, for two positively-oriented dual contours $\gamma_1,\gamma_2$ satisfying that $f$ and $g$ are discrete holomorphic on the symmetric differences $\interior_\black\gamma_1\SymDiff\interior_\black\gamma_2$ and $\interior_\white\gamma_1\SymDiff\interior_\white\gamma_2$ respectively,
one has $\sqint_{\gamma_1} f(z_\white)g(z_\black)\dd{z}=
\sqint_{\gamma_2} f(z_\white)g(z_\black)\dd{z}$.

Moreover one also has \emph{discrete integration by parts}: if $f,g:\ZWhite\longmapsto\C$ are discrete holomorphic on a discrete neighborhood of $\gamma$ it follows that
$$
\sqint_\gamma \dee f(z_\black)g(z_\white)\dd{z}
=
-\sqint_\gamma f(z_\white)\dee g(z_\black)\dd{z}\,.
$$

As for the discrete Laurent monomials, a modified version of the ones constructed in {\cite{HKV}} shall be considered\footnote{For the results presented here, the modification is not essential --- one could reproduce the rest of the results in this paper with the original choice of Laurent monomials.
Nevertheless, the monomials in {\cite{HKV}} do not allow one to build both the holomorphic and antiholomorphic sectors simultaneously with the expected commutativity among them.}.
In particular, one should distribute the discrete pole in the black sub\-lattice of $\Z^2$ among the four $\ZBlack$-nearest neighbours of the origin --- see the fifth property in the following proposition.

\begin{prop}[Proposition $2.1$ in \cite{HKV}]\label{monomials}
There exists a unique family of functions $\{z\mapsto z^{[n]}\}_{n\in\Z}$ on $\Z^2$ that satisfies the following properties:
\begin{enumerate}
	\item For all $n\in\Z$, $z\mapsto z^{[n]}$ has the same ${\pi/2}$ rotational symmetry around the origin as $z\mapsto z^n$ on $\C$.
	\item For all $z\in\Z^2$, $z^{[0]}=1$.
	\item For any $z\in\Z^2$, there exists $N\in\N$ such that $z^{[n]}=0$ for all $n\geq N$.
	\item For $n<0$, $z^{[n]}\longrightarrow 0$ as $\norm{z}\rightarrow \infty$.
	\item The first negative-power monomial satisfies
	$$
	\frac{1}{2\pi}\deebar z^{[-1]}=\frac{1}{2}\delta_{z,0}+\frac{1}{4}\sum_{\Vert w\Vert=1}\delta_{z,w}+\frac{1}{8}\sum_{w=\pm1\pm\ii}\delta_{z,w}\,.
	$$
	\item For $n\geq 0$, for all $z\in\Z^2$ $\deebar z^{[n]}=0$. For $n<0$, there exists $R>0$ such that $\deebar z^{[n]}=0$ if $\norm{z}>R$.
	\item For any $n,m\in\Z$,
	$$
	\sqint_\gamma z^{[n]}_\black z^{[m]}_\white \dd z = 2\pi\ii \delta_{n+m+1}\,,
	$$
	for any large enough positively-oriented dual contour $\gamma$ that encircles the origin.
\end{enumerate}
\end{prop}

Define then, for $n\in\Znn$, the \emph{null radius} $\NullRad{n}$ of $z\mapsto z^{[n]}$ as the largest radius $r\in\Znn$ that satisfies the condition
$$
\norm{z}\leq r
\ \ 
\Rightarrow
\ \ 
z^{[n]}=0\,.
$$
Define also, for $n\in\Z$, the \emph{singular radius} $\SingRad{n}$ of $z\mapsto z^{[n]}$ as the smallest radius $r\in\Znn$ that satisfies the condition
$$
\norm{z}>r
\ \ \ 
\Rightarrow
\ \ 
\deebar z^{[n]}=0\,.
$$
Naturally $\SingRad{n}=0$ for $n\geq0$.

\subsection{Local fields and null fields.}
From a CFT perspective, a field $F$ is an object that can be evaluated at any point $z$ on the domain of the model in question to produce a meaningful quantity $F(z)$.
Such field is said to be local if $F(z)$ depends only on a neighbourhood of $z$, and, if $w$ is another point,
$F(z)$ and $F(w)$ have the same dependence on their respective neighbourhoods.
In other words, a local field $F$ is completely determined by a translation invariant rule and its value at a distinguished point, for example, the origin $F(0)$.
\epar

As an example, in our model of symplectic fermions in the square lattice, one could interpret $\eta(0)\xi(1)$ and $\eta(0+2\ii)\xi(1+2\ii)$ as the same local field evaluated at the points $0$ and $2\ii$.
Note, however, that one should specify the domain $\Gra$ so as to know which precise objects the above are, and therefore, to make sense of such objects independently of $\Gra$, one needs to consider a more abstract construction.
\epar

With this preamble in mind, one defines a \emph{local field} of the discrete symplectic fermions on $\Z^2$ as an element of the polynomial ring 
$$
\LocFi\coloneqq\C\big[\,\hat\varphi(z)\hat\varphi(w)
\ \big\vert\ 
w,z\in\Z^2\,\big]\,,
$$
which should be heuristically interpreted as the value of such field at $0$. Moreover, one wants to think of $\hat\varphi$ as $\hat\eta$ on white vertices and $\hat\xi$ on black vertices.
For that reason, the notation $\hat\eta(w)\hat\xi(b)$ is used interchangeably with $\hat\varphi(w)\hat\varphi(b)$ when the colors of the vertices $w$ and $b$ are known to be white and black respectively; 
and similarly for pairs $\hat\xi\hat\eta$, $\hat\eta\hat\eta$ and $\hat\xi\hat\xi$.
The product on $\LocFi$ is stressed by a dot $\cdot$ whenever convenient.
\epar

\begin{ex}
	For $z,w\in\ZWhite$, the linear combinations
	$$
	\hat{\chi}^-(w)\hat{\chi}^+(z)
	\coloneqq
	\hat{\eta}(w)\dee\hat{\xi}(z)\,,
	\mspace{150mu}
	\hat{\chi}^+(w)\hat{\chi}^-(z)
	\coloneqq
	\dee\hat{\xi}(w)\hat{\eta}(z)\,,
	$$
	$$
	\hat{\chi}^+(w)\hat{\chi}^+(z)
	\coloneqq
	\dee\hat{\xi}(w)\dee\hat{\xi}(z)\,,
	\mspace{50mu}
	\textnormal{ and }
	\mspace{60mu}
	\hat{\chi}^-(w)\hat{\chi}^-(z)
	\coloneqq
	\hat{\eta}(w)\hat{\eta}(z)\,.\ 
	$$
	are local fields, where $\dee$ is as defined in Equation \eqref{derivatives}.
	\hfill$\diamond$
\end{ex}

Informally, a field involves the product and linear combination of $\varphi(z)$ for some $z\in\Z^2$ in a neighbourhood of $0$.
More precisely,  the \emph{support} of a local field $F\in\LocFi$ is defined as the smallest subset $S\subset\Z^2$ such that $F\in\C[\,\hat\varphi(z)\hat\varphi(w)\,\vert\, w,z\in S\,]\subset\LocFi$, and it is denoted by $\supp\, F$.

Given a domain $\Ver\subset\Z^2$ with no holes that induces a dimerable graph $\Gra$, the way a local field associates a random variable on $\Dimm{\Gra}$ to a point $z\in\Ver$ is through the map
$\ev_z^\Gra$, which is defined on monomials as follows:
Let $z_1,\ldots,z_{n+m}\in\Z^2$ be $n+m\in2\N$ vertices such that $n$ of them are white and $m$ of them are black, and consider the monomial $M=\hat\varphi(z_1)\hat\varphi(z_2)\cdots\hat\varphi(z_{n+m-1})\hat\varphi(z_{n+m})$.
Then, if $n\neq m$ or $z_i+z\notin\Ver$ for some $1\leq i\leq n+m$, set $\ev_z^\Gra(M)\equiv 0$.
Otherwise, let $\sigma\in\SymmGrp_{2n}$ be the permutation satisfying
$(z_{\sigma(1)},z_{\sigma(2)},\ldots,z_{\sigma(2n-1)},z_{\sigma(2n)})=(w_1,b_1,\ldots,w_n,b_n)$ for $w_1,\ldots,w_n\in\ZWhite$ and $b_1,\ldots,b_n\in\ZBlack$.
Define
$$
\ev_z^\Gra(M)\coloneqq
\begin{cases} 
\mspace{48mu}
\sign\sigma\ \eta(w_1+z)\xi(b_1+z)\cdots\eta(w_n+z)\xi(b_n+z)
\mspace{30mu}
&\textnormal{if }
z\in\ZWhite \\
(-1)^n\sign\sigma\ \eta(b_1+z)\xi(w_1+z)\cdots\eta(b_n+z)\xi(w_n+z)
\mspace{30mu}
&\textnormal{if }
z\in\ZBlack \\ 
\end{cases}\,.
$$
Note that choosing a different ordering of the vertices $w_i,b_i$ leads to the same random variable --- see its definition in Subsection \ref{paths-cycles}. 
The definition of $\ev_z^\Gra$ is completed by setting $\ev_z^{\Gra}(1)\equiv 1$ and extending linearly to the whole $\LocFi$.
\epar

\begin{figure}[h!]
	\centering
	\begin{overpic}[scale=0.787, tics=10]{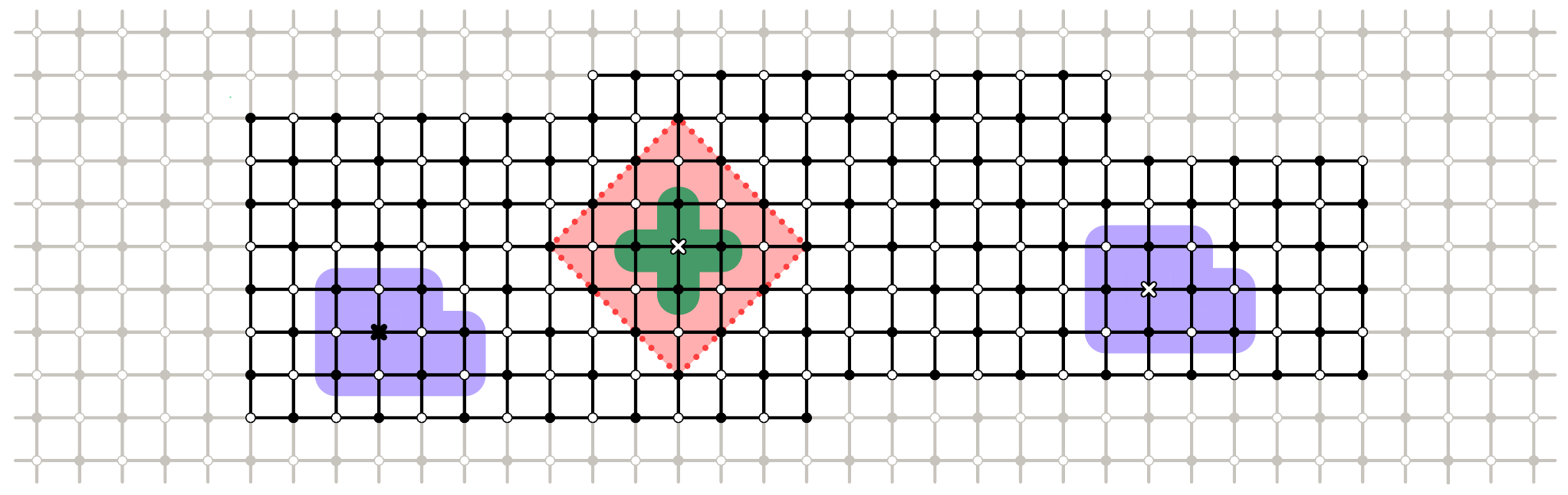}
			\put(71.1,13.7){\small $w$}
			\put(22.4,10.9){\small $b$}
			\put(13.7,24.4){{\large$\Gra$}}
	\end{overpic}
	\caption{A field (purple) evaluated on $\Gra$ at $w$ and $b$ and \\ a null field (green) with a radius of nullity (red) thereof.}
	\label{fields}
\end{figure}

From the point of view of CFT, the relevant quantities of a model are the correlation functions of fields evaluated at macroscopically separated points.
In our construction, those quantities are expectation values of local fields evaluated at such points, and for that reason, one should identify any two local fields that lead to the same expectation values when tested against local fields at a large enough distance.
This motivates the following definition of null fields.
\epar

A local field $F\in\LocFi$ is said to be \emph{null} if there exists $R>0$ such that
$$
\Expect_\Gra\big[\ev^\Gra_z\big(\,F\cdot\hat\varphi(z_1)\cdots\hat\varphi(z_{2n})\,\big)\big]=0
$$
for any domain $\Gra=(\Ver,\Edg)$ of $\Z^2$ with no holes, any $z_1,\ldots,z_{2n}\in\Z^2\setminus\Ball(0;R)$
and any $z\in\Ver$ that satisfy $\Ball(z;\rad\,F)\subset\Ver$ and $z+z_i\in\Ver$ for $1\leq i\leq 2n$.
Such an $R$ is called a \emph{radius of nullity} of $F$.

\begin{ex}\label{fields-anticommute}
	Note the evaluation map is built to encode the anticommutativity of the fermions: the local field $\hat\eta(w)\hat\xi(b)+\hat\xi(b)\hat\eta(w)$ is null for all $w\in\ZWhite$ and $b\in\ZBlack$.
	\hfill$\diamond$
\end{ex}

\begin{ex}
	By Corollary \ref{null-fermions}, local fields of the form $\hat\eta(w_1)\hat\xi(b_1)\cdots \hat\eta(w_n)\hat\xi(b_n)\in\LocFi$ with $w_i=w_j$ or $b_i=b_j$ for some $1\leq i\neq j\leq n$ are null.
	\hfill$\diamond$
\end{ex}

\begin{ex}\label{field-derivatives}
	For $w,w_0\in\ZWhite$, the field $\hat{\eta}(w_0)\deebar\hat{\xi}(w)$, where $\deebar$ is as in Equation \eqref{derivatives}, is null if and only if $w\neq w_0$ by Wick's theorem and Theorem \ref{discrete_holomorphicity}.
	Similarly, for $b,b_0\in\ZBlack$, the field $\deebar\hat{\eta}(b)\hat{\xi}(b_0)$ is null if and only if $b\neq b_0$.
	\hfill$\diamond$
\end{ex}
The set of null fields is denoted by $\NuFi\subset\LocFi$.
Note it is a vector subspace, but it is not a subalgebra:
Take the null fields $F_1=\hat{\eta}(w)\deebar\hat{\xi}(0)$ and $F_2=\hat{\eta}(0)\deebar\hat{\xi}(w)$ for some $w\in\ZWhite\setminus\{0\}$.
Then $F_1\cdot F_2$ is not null by Wick's theorem and Theorem \ref{discrete_holomorphicity}.

\begin{rmk}\label{null-rmk}
	Let $F\in\LocFi$ be a local field.
	By virtue of Wick's theorem and Theorem \ref{discrete_holomorphicity}, the fields defined in Example \ref{field-derivatives} satisfy
	$$
	\hat{\eta}(w_0)\deebar\hat{\xi}(w)\cdot F=\delta_{w,w_0}F+\NuFi
	\ \ \ \ \ \ 
	\textnormal{ and }
	\ \ \ \ \ \ 
	\deebar\hat{\eta}(b)\hat{\xi}(b_0)\cdot F= -\delta_{b,b_0}F+\NuFi
	$$
	for $w,w_0\in\ZWhite\setminus\supp_\white F$ and $b,b_0\in\ZBlack\setminus\supp_\black F$.
	Furthermore, $\deebar\hat\chi^\alpha(z)\hat\chi^\beta(w)\cdot F$ and $\hat\chi^\alpha(z)\deebar\hat\chi^\beta(w)\cdot F$ ---for $z,w$ in the right sublattice of $\Z^2$ depending on $\alpha,\beta$ in each case--- are null as long as $z\neq w$, and $z,w$ are at least at distance $2$ from $\supp F$.
	\hfill $\diamond$
\end{rmk}

\subsection{Fermionic current modes.}
Let $\gamma,\gamma^+$ be a pair of positively-oriented dual contours.
Define then, for $n,m\in\Z$ and $\alpha,\beta\in\{+,-\}$, the local field
\begin{align*}
(\hat\chi^\alpha_m\hat\chi^\beta_n)_{\gamma,\gamma^+}
\coloneqq\
\frac{1}{2\pi}\dcoint{\gamma^+}\dd w \dcoint{\gamma}\dd z \,w^{[m]}_\black\,z^{[n]}_\black\hat\chi^\alpha(w_\white^{})\hat\chi^\beta(z_\white^{})
\,.
\end{align*}
For this local field to produce something meaningful, the contours $\gamma$ and $\gamma^+$ need to be away from each other.
To that end, given two dual contours $\gamma=(p_0,\ldots,p_n)$ and $\tilde\gamma=(\tilde p_0,\ldots,\tilde p_m)$, define $\dist(\gamma,\tilde\gamma)\coloneqq \min_{0\leq i\leq n\,;\,0\leq j\leq m}\vert p_i - \tilde p_j\vert$, and for a set $S\subset\Z^2$ define $\dist(\gamma,S)\coloneqq \min_{0\leq i\leq n\,;\,x\in S}\vert p_i - x\vert$.

\begin{lemma}\label{indepe-of-path}
Let $F\in\LocFi$ be a local field and
let $\gamma_1,\gamma_1^+$ and $\gamma_2,\gamma_2^+$ be pairs of positively-oriented dual contours satisfying $\supp F\cup\Ball\big(0;\SingRad{n}\lor\SingRad{m}\big)\subset\interior\gamma_i\subset\interior\gamma_i^+$ and $\dist(\gamma_i^{},\gamma_i^+)>1$ and $\dist(\gamma_i^{},\supp F)>1$ for $i=1,2$.
Then, for $\alpha,\beta\in\{+,-\}$,
$$
(\hat\chi^\alpha_m\hat\chi^\beta_n)_{\gamma^{}_1,\gamma_1^+}\cdot F-(\hat\chi^\alpha_m\hat\chi^\beta_n)_{\gamma^{}_2,\gamma_2^+}\cdot F
\,\in\NuFi\,.
$$
\end{lemma}

\begin{figure}[b!]
	\centering
	\begin{overpic}[scale=0.757, tics=10]{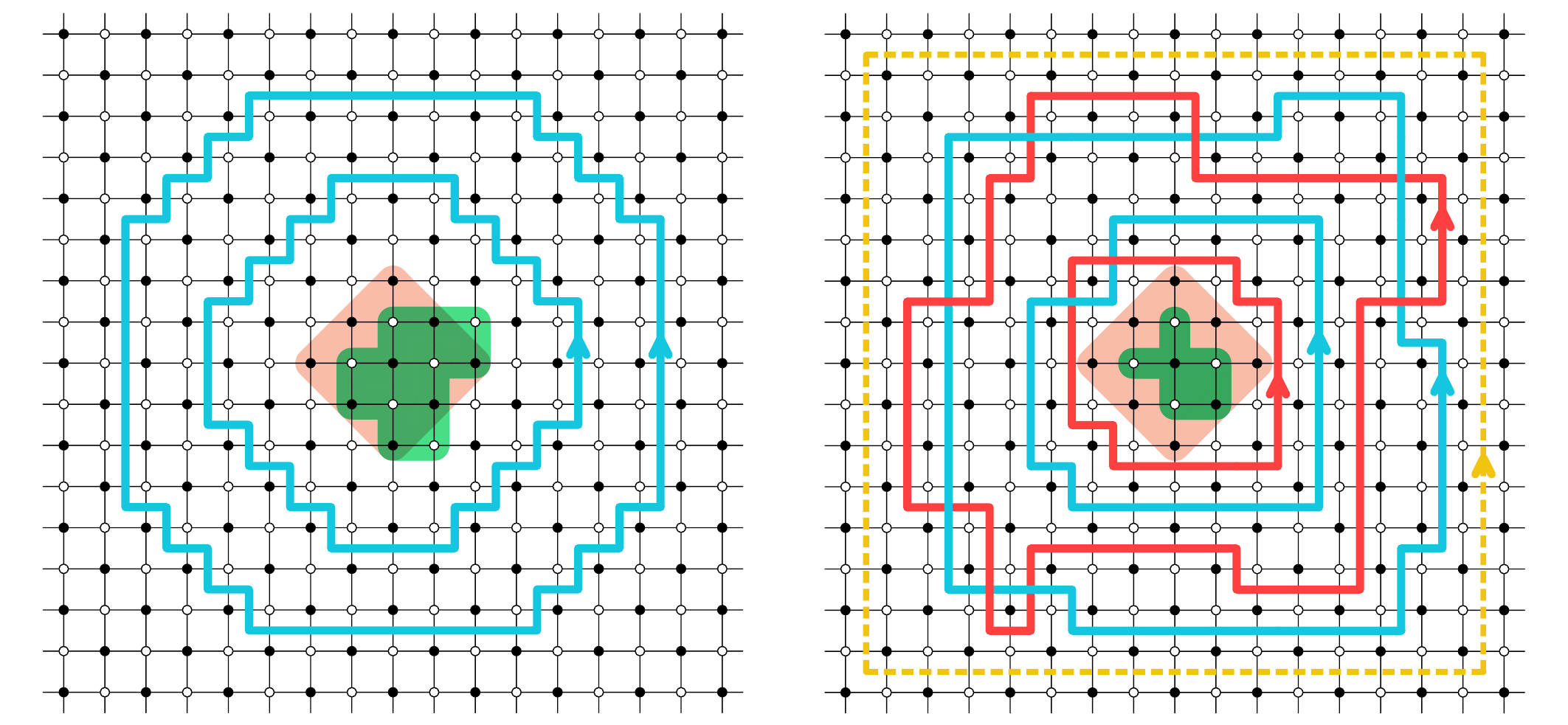}
	\end{overpic}
	\caption{\centering
		Left: Definition of $\hat\chi^\alpha_m\hat\chi^\beta_n(F+\NuFi)$.
		The support of $F$ in green, the ball $\Ball(0;\SingRad{n}\lor\SingRad{m})$ in orange,
		the contours $\gamma,\gamma^+$ in blue.
		Right: Proof of Lemma \ref{indepe-of-path}.\newline
		The contours $\gamma_1,\gamma_1^+$ in red, $\gamma_2,\gamma^+_2$ in blue and $\gamma^+$ in dashed yellow.}
	\label{fermion-modes}
\end{figure}

\begin{proof}
Take a positively-oriented dual contour $\gamma^+$ satisfying that $\gamma^+,\gamma_i^{+}$ are non-overlapping and $\interior{\gamma_i^{+}}\subset\interior{\gamma^+}$ for $i=1,2$ --- see Figure \ref{fermion-modes}.
Then, using Stokes' formula, by Remark \ref{null-rmk} and the fact that $\gamma_i$ is sufficiently away from $\supp F$ and $\gamma^+_i$,
\begin{align*}
\Big[\big(\hat\chi^\alpha_m\hat\chi^\beta_n\big)_{\gamma_i^{},\gamma_{}^+}-  \big(\hat\chi^\alpha_m\hat\chi^\beta_n\big)_{\gamma_i^{},\gamma^+_{i}}\Big]\cdot F
\,= & \ 
\frac{1}{2\pi}\dcoint{\gamma_i}\dd z\,z^{[n]}_\black
\bigg[
\dcoint{\gamma^+_{}}-\dcoint{\gamma^+_{i}}
\bigg
]\dd w \,w^{[m]}_\black\,\hat\chi^\alpha(w_\white^{})\hat\chi^\beta(z_\white^{})\cdot F
\\
= & \ 
\frac{\ii}{2\pi}\dcoint{\gamma_i}\dd z\,z^{[n]}_\black
\sum_{\underset{w_\black\notin\interior_\black\gamma^+_i}{w_\black\in\interior_\black\gamma_{}^+}}
w^{[m]}_\black\,\deebar\hat\chi^\alpha(w_\black^{})\hat\chi^\beta(z_\white^{})\cdot F
\\
= & \  0+\NuFi\,.
\end{align*}

Similar arguments prove $(\hat\chi^\alpha_m\hat\chi^\beta_n)_{\gamma_1^{},\gamma^+}\cdot F-(\hat\chi^\alpha_m\hat\chi^\beta_n)_{\gamma_2^{},\gamma^+}\cdot F$ to be null.
Writing
\begin{align*}
(\hat\chi^\alpha_m\hat\chi^\beta_n)_{\gamma_1^{},\gamma_1^+} \cdot & F
-
(\hat\chi^\alpha_m\hat\chi^\beta_n)_{\gamma_2^{},\gamma^+_2}\cdot F
\ = \  
\Big[
(\hat\chi^\alpha_m\hat\chi^\beta_n)_{\gamma_1^{},\gamma_1^+}\cdot F-(\hat\chi^\alpha_m\hat\chi^\beta_n)_{\gamma_1^{},\gamma^+}\cdot F\,
\Big]
\\
\ + & \ 
\Big[
(\hat\chi^\alpha_m\hat\chi^\beta_n)_{\gamma_1^{},\gamma^+}\cdot F-(\hat\chi^\alpha_m\hat\chi^\beta_n)_{\gamma_2^{},\gamma^+}\cdot F\,
\Big]
+ 
\Big[
(\hat\chi^\alpha_m\hat\chi^\beta_n)_{\gamma_2^{},\gamma^+}\cdot F-(\hat\chi^\alpha_m\hat\chi^\beta_n)_{\gamma_2^{},\gamma_2^+}\cdot F\,
\Big]\,,
\end{align*}
it becomes clear that the claim holds true.
\end{proof}

\begin{lemma}\label{null-to-null}
	Let $F\in\NuFi$ be a null field and
	let $\gamma,\gamma^+$ be a pair of positively-oriented dual contours satisfying $\supp F\cup\Ball(0; \SingRad{n}\,\lor\SingRad{m})\subset\interior\gamma\subset\interior\gamma^+$ and $\dist(\gamma,\gamma^+)>1$ and $\dist(\gamma,\supp F)>1$.
	Then, $(\hat\chi^\alpha_m\hat\chi^\beta_n)_{\gamma,\gamma^+}\cdot F$ is null too for any $\alpha,\beta\in\{+,-\}$.
\end{lemma}
\begin{proof}
	Let $R_F$ be a radius of nullity of $F$, and take any $R>0$ big enough such that there exists a pair $\tilde{\gamma},\tilde{\gamma}^+$ of positively-oriented dual contours satisfying
	$$
	\Ball(0;R_F+2)\subset\interior\tilde\gamma\subset\interior\tilde\gamma^+\subset\Ball(0;R-2)\,,
	$$
	and $\dist(\tilde\gamma,\tilde\gamma^+)>1$.
	By nullity of $F$, it is clear that $R$ is a radius of nullity of $(\hat\chi^\alpha_m\hat\chi^\beta_n)_{\tilde\gamma^+,\tilde\gamma}\cdot F$.
	By Lemma \ref{indepe-of-path},
	$(\hat\chi^\alpha_m\hat\chi^\beta_n)_{\gamma^+,\gamma}\cdot F-(\hat\chi^\alpha_m\hat\chi^\beta_n)_{\tilde\gamma^+,\tilde\gamma}\cdot F$ is null too, and so the claim follows.
\end{proof}

The two lemmas above are all one needs so as to check the well-definedness of the following linear operators on the space of correlation-equivalent local fields 
$$
\Fields\coloneqq\LocFi/\NuFi\,.
$$
Let $F\in\LocFi$ be a local field.
For $n,m\in\Z$, define
$$
\hat\chi^\alpha_m\hat\chi^\beta_n(F+\NuFi)
\coloneqq
(\hat\chi^\alpha_m\hat\chi^\beta_n)_{\gamma,\gamma^+}\cdot F+\NuFi\,,
$$
where $\gamma,\gamma^+$ are positively-oriented dual contours satisfying ---see Figure \ref{fermion-modes}---
$$
\supp F\cup\Ball(0;\SingRad{n}\,\lor\SingRad{m})\subset\interior\gamma\subset\interior\gamma^+\,,
$$
and $\dist(\gamma,\gamma^+)>1$.
These operators are referred to as the \emph{current modes} of the fermions.

\subsection{Anticommutation relations.} Take the standard definition of the anticommutators of the current modes:
$$
\big\{\hat\chi^\alpha_n,\hat\chi^\beta_m\big\}
\coloneqq
\hat\chi^\alpha_n\hat\chi^\beta_m+\hat\chi^\beta_m\hat\chi^\alpha_n\,,
$$
$$
\hat\chi^\gamma_k\big\{\hat\chi^\alpha_{n\phantom{l}},\hat\chi^\beta_{m\phantom{l}}\big\}\hat\chi^\delta_l
\coloneqq
\hat\chi^\gamma_k\hat\chi^\alpha_{n\phantom{l}}\circ\hat\chi^\beta_{m\phantom{l}}\hat\chi^\delta_l+\hat\chi^\gamma_k\hat\chi^\beta_{m\phantom{l}}\circ\hat\chi^\alpha_{n\phantom{l}}\hat\chi^\delta_l\,,
$$
which are again linear operators on $\Fields$.
For $x\in\Z$, let $\delta_x\coloneqq\delta_{x,0}$.
Let $\ud^{\alpha\beta}$ be an antisymmetric symbol with $\ud^{-+}=1$.
\begin{prop}\label{anticommutation}
	The fermion modes satisfy, for all $n,m\in\Z$,
	$$
	\big\{\hat\chi^\alpha_n,\hat\chi^\beta_m\big\}
	=n\,\delta_{n+m}\ud^{\alpha\beta}\id_{\Fields}\,,
	$$
	and $\hat\chi^\gamma_k\{\hat\chi^\alpha_n,\hat\chi^\beta_m\}\hat\chi^\delta_l
	=n\,\delta_{n+m}\ud^{\alpha\beta}\hat\chi^\gamma_k\hat\chi^\delta_l$.
\end{prop}
\begin{proof}
Let us consider only the case $\{\hat\chi^-_n,\hat\chi^+_m\}$ --- the rest of cases are proven with similar computations.
Fix $F\in\LocFi$ and take three non-intersecting positively-oriented dual contours $\gamma^-,\gamma,\gamma^+$ satisfying $\supp F\cup\Ball(0;\SingRad{n}\lor\SingRad{m})\subset\interior\gamma^-\subset\interior\gamma\subset\interior\gamma^+$, and $\dist(\supp F,\gamma^-)>1$ and $\dist(\gamma,\gamma^\pm)>1$.
Recall from Example \ref{fields-anticommute} that $\hat\eta(w)\hat\xi(b)\cdot F=-\hat\xi(b)\hat\eta(w)\cdot F+\NuFi$ for any $F\in\LocFi$.
Then,
the following chain of equalities holds modulo null fields,
by Stokes' formula, the integration by parts formula, Remark \ref{null-rmk}, and the integral properties of the discrete monomials ---Proposition \ref{monomials}---:
\begin{align*}
\Big[
\big(\hat\chi^-_n\hat\chi^+_m\big)_{\gamma_-,\gamma} +\big(\hat\chi^+_m\hat\chi^-_n\big)_{\gamma,\gamma_+}
\Big]\cdot F
\ = & \ 
\frac{1}{2\pi}\sqint_\gamma\dd w \sqint_{\gamma_-}\dd z\, w^{[n]}_\black\,z^{[m]}_\black\,
\hat\eta(w_\white^{})\dee\hat\xi(z_\white^{})\cdot F
\\
\ \ & \mspace{40mu} +
\frac{1}{2\pi}\sqint_{\gamma_+}\dd z \sqint_{\gamma}\dd w\, w^{[n]}_\black\,z^{[m]}_\black\,
\dee\hat\xi(z_\white^{})\hat\eta(w_\white^{})\cdot F
\\
\ = & \ 
\frac{1}{2\pi}
\sqint_\gamma\dd w\,w^{[n]}_\black
\bigg[\sqint_{\gamma^-}-\sqint_{\gamma^+}\bigg]\dd z
\,z^{[m]}_\black\,\hat\eta(w_\white^{})\dee\hat\xi(z_\white^{})\cdot F
\\
\ = & \ -
\frac{m}{2\pi}
\sqint_\gamma\dd w\,w^{[n]}_\black
\bigg[\sqint_{\gamma^-}-\sqint_{\gamma^+}\bigg]\dd z
\,z^{[m-1]}_\white\,\hat\eta(w_\white^{})\hat\xi(z_\black^{})\cdot F
\\
\ = & \ 
\frac{\ii m}{2\pi}
\sqint_\gamma\dd w\,w^{[n]}_\black
\sum_{\underset{z_\white\notin\interior_\white\gamma^-}{z_\white\in\interior_\white\gamma^+}}
\,z^{[m-1]}_\white\,\hat\eta(w_\white^{})\deebar\hat\xi(z_\white^{})\cdot F
\\
\ = & \ 
\frac{\ii m}{2\pi}
\sqint_\gamma\dd w\,w^{[n]}_\black
\,w^{[m-1]}_\white\,F
\\
\ = & \ 
-m\delta_{n+m}\,F
\\
\ = & \ 
n\delta_{n+m}\,F\,,
\end{align*}
which completes the proof for $\alpha=-$ and $\beta=+$.
\end{proof}

%{\color{blue}
%	(Might not need this after all.)
%	Define the \emph{fermion algebra} on $\Gra=(\White\sqcup\Black,\Edg)$ as the vector subspace 
%	$$
%	\FerAlg{\Gra}
%	\coloneqq
%	\textnormal{span}_\C\big\{\eta(w_1)\xi(b_1)\cdots\eta(w_n)\xi(b_n)\,\big\vert\, w_1,\ldots,w_n\in\ZWhite,\,b_1,\ldots,b_n\in\ZBlack\big\}
%	$$
%	of random variables on $\Dimm{\Gra}$, where $\eta(w_1)\xi(b_1)\cdots\eta(w_n)\xi(b_n)\equiv 0$ when some $w_i$ or $b_i$ sits outside of $\White\sqcup\Black$.
%	Then $\FerAlg{\Gra}$ is equipped with the as\-so\-cia\-tive product $\FerProd\colon\FerAlg{\Gra}\times\FerAlg{\Gra}\longrightarrow\FerAlg{\Gra}$ defined on monomial elements by
%	$$
%	\big[\eta(w_1)\xi(b_1)\cdots\eta(w_n)\xi(b_n)\big]\wedge\big[\eta(w_{n+1})\xi(b_{n+1})\cdots\eta(w_m)\xi(b_m)\big]
%	\coloneqq
%	$$
%	$$
%	\eta(w_1)\xi(b_1)\cdots\eta(w_n)\xi(b_n)\eta(w_{n+1})\xi(b_{n+1})\cdots\eta(w_m)\xi(b_m)\,,
%	$$
%	and extended bilinearly. 
%	{\red
%		The \emph{antifermion algebra} $\AntiFerAlg{\Gra}$ is defined similarly.
%	}
%}

\newpage
\section{Virasoro representation on local fields}
\label{sec: virasoro}

The CFT of symplectic fermions \cite{Kausch2} suggests that, from the fermion current modes, one can build operators satisfying the Virasoro relations --- see Section \ref{sec:intro}.
Doing so, the space of local fields $\Fields$ can be rendered a representation of the Virasoro algebra.
In particular, the Virasoro modes are built through a \emph{Sugawara construction}, i.e. as an  infinite sum that is quadratic on the fermion modes. 
In this context, the Sugawara construction is justified by the following truncation property.
 
\begin{lemma}\label{truncation}
	Let $F\in\LocFi$ be a local field.
	There exists $N\in\N$ such that
	$$
	\hat\chi^\alpha_m\hat\chi^\beta_n(F+\NuFi)=0+\NuFi
	$$
	for all $n\geq N$ and all $m\in\Z$ and $\alpha,\beta\in\{+,-\}$.
\end{lemma}
\begin{proof}
	Fix $m\in\Z$.
	Take $N\in\N$ large enough such that there exists a positively-oriented dual contour $\gamma$ satisfying $\supp F\subset\interior\gamma\subset\Ball(0;\NullRad{N}-1)$.
	Such a contour always exists since $\NullRad{n}$ grows with $n$ --- see Proposition \ref{monomials}.
	Then, clearly $(\hat\chi^\alpha_m\hat\chi^\beta_n)_{\gamma,\gamma^+}\cdot F=0$ for any $\gamma^+$ since the in\-te\-grand $z\mapsto z^{[n]}$ vanishes everywhere along $\gamma$.
\end{proof}

\subsection{Virasoro modes.}
For $n\in\Z$, the \emph{Virasoro modes} are defined as the linear operators
$$
\mathbf{L}_n\coloneqq\sum_{k\in\Z}\,\NormOrd{0}{\hat\chi^+_{n-k}\hat\chi^-_k}\,
$$
on $\Fields$, where the \emph{normally ordered} fermion modes are defined as
$$
\NormOrd{k}{\hat\chi^+_m\hat\chi^-_n}\,
\coloneqq
\begin{cases} 
\mspace{17mu}
{\hat\chi^+}_m {\hat\chi^-}_n
\mspace{20mu}
&
\textnormal{if}
\mspace{10mu}
n-m\geq k
\\
\,-{\hat\chi^-}_n {\hat\chi^+}_m 
\mspace{20mu}
&
\textnormal{if}
\mspace{10mu}
n-m< k
\end{cases}\,.
$$
Note that, for all $n\in\Z$, the operator $\mathbf{L}_n$ on $\Fields$ is well-defined since, by Lemma \ref{truncation}, only finitely many terms yield a non-zero field when $\mathbf{L}_n$ acts on a given field.
\epar

The following remark is used later in some proofs.

\begin{rmk}\label{normal-ord}
	For $k<l$,
	$$
	\NormOrd{k}{\hat\chi^+_m\hat\chi^-_n}-\NormOrd{l}{\hat\chi^+_m\hat\chi^-_n}
	=
	\begin{cases} 
	\big\{\hat\chi^+_m,\hat\chi^-_n\big\}
	\mspace{20mu}
	&
	\textnormal{if}
	\mspace{10mu}
	k\leq n-m<l
	\\
	\mspace{30mu}
	0
	\mspace{20mu}
	&
	\textnormal{otherwise}
	\end{cases}\,,
	$$
	for all $n,m\in\Z$.\hfill$\diamond$
\end{rmk}

\subsection{Central charge $-2$.}
The main result in this section states that the $\mathbf{L}_n$ satisfy the Virasoro commutation relations with central charge $-2$.
Let $[\,\cdot\,,\cdot\,]$ denote the usual commutator of operators, i.e. $[A\,,B]\coloneqq A\circ B-B\circ A$.

\begin{thm}\label{Virasoro}
	For $n,m\in\Z$
	$$
	\big[\mathbf L_n,\mathbf L_m\big]
	=
	(n-m)\mathbf L_{n+m}+\frac{c}{12}(n^3-n)\delta_{n+m}\id_{\Fields}
	$$
	with $c=-2$.
\end{thm}

The proof of this theorem is broken down into three separate lemmas.

\begin{lemma}\label{comm-1step}
	For any $n,m,l,k\in\Z$, 
	$$
	\big[\mathbf L_n,\NormOrd{m}{\hat\chi^+_l\hat\chi^-_k}\big]
	=
	-k\,\NormOrd{-(l+k)\phantom{-}}{\hat\chi^+_l\hat\chi^-_{n+k}}-l\,\NormOrd{l+k}{\hat\chi^+_{n+l}\hat\chi^-_{k}}\,.
	$$
\end{lemma}
\begin{proof}
	Note that,
	in the following chain of equalities,
	all infinite sums of operators produce a well-defined operator by Lemma \ref{truncation}.
	In particular, this property justifies pulling the infinite sum out of commutators.
	Then, one just needs to use the identity $[AB,CD]=A\{B,C\}D+\{A,C\}\circ DB-C\{A,D\}B-AC\circ\{B,D\}$:
	\begin{align*}
	\big[\mathbf L_n,\hat\chi^+_l\hat\chi^-_k\big]= &
	\sum_{j\in\Z}\big[\NormOrd{0}{\hat\chi^+_{n-j}\hat\chi^-_j},\hat\chi^+_l\hat\chi^-_k\big]
	\\
	= &
	\sum_{j\geq n/2} \big[{\hat\chi^+_{n-j}\hat\chi^-_j},\hat\chi^+_l\hat\chi^-_k\big]
	-
	\sum_{j< n/2} \big[{\hat\chi^-_j\hat\chi^+_{n-j}},\hat\chi^+_l\hat\chi^-_k\big]
	\\
	= & \sum_{j\geq n/2} \Big(\hat\chi^+_{n-j}\{\hat\chi^-_j,\hat\chi^+_{l}\}\hat\chi^-_k
	-
	\hat\chi^+_{l}\{\hat\chi^-_k,\hat\chi^+_{n-j}\}\hat\chi^-_j
	\Big)
	\\
	& \mspace{80mu} -
	\sum_{j< n/2} \Big(\{\hat\chi^-_j,\hat\chi^+_{l}\}\circ\hat\chi^-_k\hat\chi^+_{n-j}
	- 
	\hat\chi^-_j\hat\chi^+_{l}\circ\{\hat\chi^-_k,\hat\chi^+_{n-j}\}
	\Big)
	%		\\
	%		= &
	%		\,l\,\hat\chi^-_k\hat\chi^+_{n+l}\indicator{\{l\leq -n/2\}}
	%		+
	%		(n+k)\hat\chi^-_{n+k}\hat\chi^+_l\indicator{\{k\geq -n/2\}} -
	%		\\
	%		& \mspace{150mu} -
	%		l\,\hat\chi^+_{n+l}\hat\chi^-_k\indicator{\{l> -n/2\}}
	%		-
	%		(n+k)\hat\chi^+_l\hat\chi^-_{n+k}\indicator{\{k< -n/2\}}
	\\
	= &\  -l\,\NormOrd{l+k}{\hat\chi^+_{n+l}\hat\chi^-_k}-k\,\NormOrd{-(l+k)\phantom{-}}{\hat\chi^+_l\hat\chi^-_{n+k}}\,,
	\end{align*}
	where in the last step one needs to observe that, for fixed values of $l,k\in\Z$, two and only two of the four terms produce a non-zero operator.
	A similar computation yields $[\mathbf{L}_n,\hat\chi^-_k\hat\chi^+_l]=-[\mathbf{L}_n,\hat\chi^+_l\hat\chi^-_k]$, and the claim follows.
\end{proof}

Define $\odd,\even:\Zpos\longrightarrow\{0,1\}$ by $\odd(n)=1$ if $n$ is odd and $0$ otherwise, and $\even(n)\coloneqq1-\odd(n)$.
Define also $\Theta:\Z\longrightarrow\{0,1\}$ by $\Theta(n)=1$ if $n\geq 0$ and $0$ otherwise.

\begin{lemma}\label{trick1}
	For any $n\in\Z$ and $m\in\Z\setminus\{0\}$,
	$$
	\sum_{k\in\Z}\NormOrd{m}{\hat\chi^+_{n+m-k}\hat\chi^-_k}
	=
	\mathbf{L}_{n+m}
	-
	\delta_{n+m}\bigg(
	\frac{|m|}{2}\even(|m|)\Theta(-m)
	+
	\sum_{0< k<\frac{|m|}{2}}k
	\bigg)\id_{\Fields}
	\,.
	$$
\end{lemma}
\begin{proof}
	Again, the infinite sums are well-defined by Lemma \ref{truncation}.
	From Remark \ref{normal-ord} and Proposition \ref{anticommutation} it follows
	$$
	\sum_{k\in\Z}\NormOrd{m}{\hat\chi^+_{n+m-k}\hat\chi^-_k}
	=
	\sum_{k\in\Z}\NormOrd{0}{\hat\chi^+_{n+m-k}\hat\chi^-_k}
	+\delta_{n+m}\id_{\Fields}\cdot
	\begin{cases}
	\mspace{17mu}
	\sum_{\frac{m}{2}\leq k<0}k
	\mspace{10mu}
	&
	\textnormal{if}
	\mspace{10mu}
	m<0
	\\
	- 
	\sum_{0\leq k<\frac{m}{2}}k
	\mspace{10mu}
	&
	\textnormal{if}
	\mspace{10mu}
	m>0
	\end{cases}\,,
	$$
	from which the claim follows immediately.
\end{proof}

\begin{lemma}\label{trick2}
	For any $n\in\Z$ and $m\in\Z\setminus\{0\}$,
	$$
	\sum_{k\in\Z}k\,\bigg(\NormOrd{m}{\hat\chi^+_{n+m-k}\hat\chi^-_k}-\NormOrd{-m\phantom{-}}{\hat\chi^+_{n+m-k}\hat\chi^-_k}\bigg)
	=
	-\bigg(\frac{m^3-m}{12}+\frac{m}{4}\even(|m|)\bigg)\delta_{n+m}\id_{\Fields}\,.
	$$
\end{lemma}
\begin{proof}
	Again, the infinite sums are well-defined by Lemma \ref{truncation}.
	From Remark \ref{normal-ord} and Proposition \ref{anticommutation} it follows
	$$
	\sum_{k\in\Z}k\,\bigg(\NormOrd{m}{\hat\chi^+_{n+m-k}\hat\chi^-_k}-\NormOrd{-m\phantom{-}}{\hat\chi^+_{n+m-k}\hat\chi^-_k}\bigg)
	=
	\delta_{n+m}\id_{\Fields}\cdot
	\begin{cases}
	\mspace{17mu}
	\sum_{\frac{m}{2}\leq k<-\frac{m}{2}}k^2
	\mspace{10mu}
	&
	\textnormal{if}
	\mspace{10mu}
	m<0
	\\
	- 
	\sum_{-\frac{m}{2}\leq k<\frac{m}{2}}k^2
	\mspace{10mu}
	&
	\textnormal{if}
	\mspace{10mu}
	m>0
	\end{cases}\,,
	$$
	from which the claim follows immediately.
	\end{proof}

\begin{proof}[Proof of Theorem \ref{Virasoro}]
Using Lemma \ref{comm-1step},
\begin{align*}
\big[\mathbf L_n,\mathbf L_m\big]
& =
\sum_{k\in\Z}\big[\mathbf L_n,\NormOrd{0}{\hat\chi^+_{m-k}\hat\chi^-_k}\big]
=
-\sum_{k\in\Z}\bigg(k\,\NormOrd{-m\phantom{-}}{\hat\chi^+_{m-k}\hat\chi^-_{n+k}}+(m-k)\NormOrd{m}{\hat\chi^+_{n+m-k}\hat\chi^-_{k}}\bigg)
\\
& = -
\sum_{k\in\Z}\bigg((k-n)\NormOrd{-m\phantom{-}}{\hat\chi^+_{m+n-k}\hat\chi^-_{k}}+(m-k)\NormOrd{m}{\hat\chi^+_{n+m-k}\hat\chi^-_{k}}\bigg)
\\
& = 
n\,\sum_{k\in\Z}\NormOrd{-m\phantom{-}}{\hat\chi^+_{m+n-k}\hat\chi^-_{k}}
-
m\,\sum_{k\in\Z}\NormOrd{m}{\hat\chi^+_{n+m-k}\hat\chi^-_{k}}
+
\sum_{k\in\Z}k\,\bigg(\NormOrd{m}{\hat\chi^+_{n+m-k}\hat\chi^-_{k}}-\NormOrd{-m\phantom{-}}{\hat\chi^+_{m+n-k}\hat\chi^-_{k}}\bigg)\,.
\end{align*}
At this point, it is clear that for $m=0$, $[\mathbf L_n,\mathbf L_0]=n\mathbf{L}_{n}$.
For $m\neq0$, using Lemmas \ref{trick1} and \ref{trick2} on the above expression,
\begin{align*}
\big[\mathbf L_n,\mathbf L_m\big]
& =
(n-m)\mathbf{L}_{n+m}
-
n\,\delta_{n+m}\bigg(
\frac{|m|}{2}\even(|m|)\Theta(m)
+
\sum_{0< k<\frac{|m|}{2}}k
\bigg)\id_{\Fields}
\\
& \mspace{80mu}
+m\,
\delta_{n+m}\bigg(
\frac{|m|}{2}\even(|m|)\Theta(-m)
+
\sum_{0< k<\frac{|m|}{2}}k
\bigg)\id_{\Fields}
\\
& \mspace{80mu}
-\delta_{n+m}\bigg(\frac{m^3}{12}+\frac{m}{6}-\frac{m}{4}\odd(m)\bigg)\id_{\Fields}\,.
\end{align*}
The proof is finally complete using $\sum_{0< k<\frac{|m|}{2}}k=(m^2-2|m|\even(|m|)-\odd(|m|))/8$ and $\Theta(m)+\Theta(-m)=1$ for $m\neq0$.
%$$
%\big[\mathbf L_n,\mathbf L_m\big]-(n-m)\mathbf L_{n+m}
%=
%-2\cdot \frac{m^3-m}{12}
%\delta_{n+m}\id_{\Fields}\,.
%$$
\end{proof}

\newpage
\titleformat{\section}
{\normalfont\Large\bfseries}{\thesection}{0pt}{}

\end{document}